\documentclass[journal]{IEEEtran}
\usepackage{amsfonts}
\usepackage{cite,graphicx,amsmath,amsthm}
\usepackage{subfigure}
\usepackage{fancyhdr}
\usepackage{dsfont}
\usepackage{array,color}
\usepackage{bm}
\usepackage{float}
\usepackage{algorithm}
\usepackage{algpseudocode}
\usepackage{multirow}
\usepackage{booktabs}
\usepackage{multirow}
\usepackage{makecell}

\newtheorem{theorem}{Theorem}

\newtheorem{lemma}{Lemma}

\newtheorem{proposition}{Proposition}

\newtheorem{corollary}{Corollary}

\newtheorem{property}{Property}

\newtheorem{remark}{Remark}

\newtheorem{claim}{Claim}

\allowdisplaybreaks[4]

\begin{document}
\title{{Machine Intelligence at the Edge\\with Learning Centric Power Allocation}}

\author{Shuai~Wang,~\IEEEmembership{Member,~IEEE},
Yik-Chung~Wu,~\IEEEmembership{Senior~Member,~IEEE},
Minghua~Xia,~\IEEEmembership{Member,~IEEE},
Rui~Wang,~\IEEEmembership{Member,~IEEE},
and~H.~Vincent~Poor,~\IEEEmembership{Life~Fellow,~IEEE}

\thanks{

This work was supported in part by the National Natural Science Foundation of China under Grants 61771232 and 61671488, in part by the Shenzhen Basic Research Project under Grant JCYJ20190809142403596,
in part by the Natural Science Fundation of Guangdong Province under Grant 2019B1515130003, in part by the U.S. National Science Foundation under Grants CCF-0939370 and CCF-1908308, in part by the Major Science and Technology Special Project of Guangdong Province under Grant 2018B010114001, and in part by the Fundamental Research Funds for the Central Universities under Grant 191gjc04. (\textit{Corresponding Author: Yik-Chung Wu}).

Shuai Wang is with the Department of Electrical and Electronic Engineering, and also with the Department of Computer Science and Engineering, Southern University of Science and Technology, Shenzhen 518055, China (e-mail: wangs3@sustech.edu.cn).

Yik-Chung Wu is with the Department of Electrical and Electronic Engineering, The University of Hong Kong, Hong Kong (e-mail: ycwu@eee.hku.hk).

Minghua Xia is with the School of Electronics and Information Technology, Sun Yat-sen University, Guangzhou 510006, China, and also with the Southern Marine Science and Engineering Guangdong Laboratory, Zhuhai 519082, China (e-mail: xiamingh@mail.sysu.edu.cn).

Rui Wang is with the Department of Electrical and Electronic Engineering, Southern University of Science and Technology, Shenzhen 518055, China (e-mail: wang.r@sustech.edu.cn).

H. Vincent~Poor is with the Department of Electrical Engineering, Princeton University, Princeton, NJ 08544 USA (e-mail: poor@princeton.edu).

}
}

\maketitle

\begin{abstract}
While machine-type communication (MTC) devices generate considerable amounts of data, they often cannot process the data due to limited energy and computational power.
To empower MTC with intelligence, edge machine learning has been proposed.
However, power allocation in this paradigm requires maximizing the learning performance instead of the communication throughput, for which the celebrated water-filling and max-min fairness algorithms become inefficient.
To this end, this paper proposes learning centric power allocation (LCPA), which provides a new perspective on radio resource allocation in learning driven scenarios.
By employing 1) an empirical classification error model that is supported by learning theory and 2) an uncertainty sampling method that accounts for different distributions at users, LCPA is formulated as a nonconvex nonsmooth optimization problem, and is solved using a majorization minimization (MM) framework.
To get deeper insights into LCPA, asymptotic analysis shows that the transmit powers are inversely proportional to the channel gains, and scale exponentially with the learning parameters.
This is in contrast to traditional power allocations where quality of wireless channels is the only consideration.
Last but not least, a large-scale optimization algorithm termed mirror-prox LCPA is further proposed to enable LCPA in large-scale settings.
Extensive numerical results demonstrate that the proposed LCPA algorithms outperform traditional power allocation algorithms, and the large-scale optimization algorithm reduces the computation time by orders of magnitude compared with MM-based LCPA but still achieves competing learning performance.
\end{abstract}


\IEEEpeerreviewmaketitle
\section{Introduction}

\IEEEPARstart{M}ACHINE intelligence is revolutionizing every branch of science and technology \cite{ml1,dl1}.
If a machine wants to learn, it requires at least two ingredients: information and computation, which are usually separated from each other in machine-type communication (MTC) systems \cite{iot1}.
Nonetheless, sending vast volumes of data from MTC devices to the cloud not only leads to a heavy communication burden but also increases the transmission latency.
To address this challenge brought by MTC, a promising solution is the \emph{edge machine learning} technique \cite{edge1,edge2,edge3,edge4,fed1,fed2,fed3,fed4} that trains a machine learning model or \emph{fine-tunes a pre-trained model} at the edge, i.e., at a nearby radio access point with computation resources.

In general, there are two ways to implement edge machine learning: data sharing and model sharing.
Data sharing uses the edge to collect data generated from MTC devices for machine learning \cite{edge1,edge2,edge3,edge4}, while model sharing uses federated learning \cite{fed1,fed2,fed3,fed4} to exchange model parameters (instead of data) between the edge and users.
Both approaches are recognized as key paradigms in the sixth generation (6G) wireless communications \cite{6g1,6g2,6g3}.
However, since the MTC devices often cannot process the data due to limited computational power, this paper focuses on data sharing.

\subsection{Motivation and Related Work}

In contrast to conventional communication systems, edge machine learning systems aim to maximize the learning performance instead of the communication throughput.
Therefore, edge resource allocation becomes very different from traditional resource allocation schemes that merely consider the wireless channel state information \cite{waterfilling,fair,sumrate,gradient}.
For instance, the celebrated water-filling scheme allocates more resources to better channels for throughput maximization \cite{waterfilling}, and the max-min fairness scheme allocates more resources to cell-edge users to maintain certain quality of service \cite{fair}.
While these two schemes have proven to be very efficient in traditional wireless communication systems, they could lead to poor learning performance in edge learning systems because they do not account for the machine learning factors such as model and dataset complexities.
Imagine training a deep neural network (DNN) and a support vector machine (SVM) at the edge.
Due to much larger number of parameters in DNN, the edge should allocate more resources to MTC devices that upload data for the DNN than those for the SVM.

Nonetheless, in order to maximize the learning performance, we need a mathematical expression of the learning performance with respect to the number of samples, which does not exist to the best of the authors' knowledge.
While the sample complexity of a learning task can be related to the Vapnik-Chervonenkis (VC) dimension \cite{ml1}, this theory only provides a vague estimate that is independent of the specific learning algorithm or data distribution.
To better understand the learning performance, it has been proved in \cite{model3,model4} that the generalization error can be upper bounded by the summation of the bias between the model's prediction and the optimal prediction, the variance due to training datasets, and the noise of the target example.
With the bound being tight for certain loss functions (e.g., squared loss and zero-one loss), the bias-variance decomposition theory gives rise to an empirical nonlinear classification error model \cite{model1,model2,model5} that is also theoretically supported by the inverse power law derived via statistical mechanics \cite{model6}.

\subsection{Summary of Results}

In this paper, we adopt the above nonlinear model to approximate the learning performance, and a \emph{learning centric power allocation (LCPA)} problem is formulated with the aim of minimizing classification error subject to the total power budget constraint.
When the data is non-independent-and-identically-distributed (non-IID) among users, the LCPA problem can also be flexibly integrated with uncertainty sampling.
Since the formulated machine learning resource allocation problem is nonconvex and nonsmooth, it is nontrivial to solve.
By leveraging the majorization minimization (MM) framework from optimization, an MM-based LCPA algorithm that obtains a Karush-Kuhn-Tucker (KKT) solution is proposed.
To get deeper insights into LCPA, asymptotic analysis with the number of antennas at the edge going to infinity is provided.
The asymptotic optimal solution discloses that the transmit powers are inversely proportional to the channel gain, and scale exponentially with the classification error model parameters.
This result reveals that machine learning has a stronger impact than wireless channels in LCPA.
To enable affordable computational complexity when the number of MTC devices is large, a variant of LCPA, called mirror-prox LCPA, is proposed.
The algorithm is a first-order method (FOM), implying that its complexity is linear with respect to the number of users.
Extensive experimental results based on public datasets show that the proposed LCPA scheme is able to achieve a higher classification accuracy than that of the sum-rate maximization and max-min fairness power allocation schemes.
For the first time, the benefit brought by joint communication and learning design is quantitatively demonstrated in edge machine learning systems.
Our results also show that the mirror-prox LCPA reduces the computation time by orders of magnitude compared to the MM-based LCPA but still achieves satisfactory performance.

To sum up, the contributions of this paper are listed as follows.
\begin{itemize}
\item The LCPA scheme is developed for the edge machine learning problem, which maximizes the learning accuracy instead of the communication throughput.

\item To understand how LCPA works, an asymptotic optimal solution to the edge machine learning problem is derived, which, for the first time, discloses that the transmit power obtained from LCPA grows linearly with the path-loss and grows exponentially with the learning parameters.

\item To reduce the computation time of LCPA in the massive multiple-input multiple-output (MIMO) setting, a variant of LCPA based on FOM is proposed, which enables the edge machine learning system to scale up the number of MTC users.

\item Extensive experimental results based on public datasets (e.g., MNIST, CIFAR-10, ModelNet40) show that the proposed LCPA is able to achieve a higher accuracy than that of the sum-rate maximization and max-min fairness schemes.

\end{itemize}

\subsection{Outline}

The rest of this paper is organized as follows.
System model and problem formulation are described in Section II.
Classification error modeling is presented in Section III.
The MM-based LCPA algorithm, the asymptotic analysis, and the large-scale optimization method are derived in Sections IV, V, and VI, respectively.
Finally, experimental results are presented in Section VII, and conclusions are drawn in Section~VIII.

\emph{Notation}:
Italic letters, lowercase and uppercase bold letters represent scalars, vectors, and matrices, respectively.
Curlicue letters stand for sets and $|\cdot|$ is the cardinality of a set.
The operators $(\cdot)^{T}$, $(\cdot)^{H}$ and $(\cdot)^{-1}$ take the transpose, Hermitian and inverse of a matrix, respectively.
We use $(a_1,a_2,\cdots)$ to represent a sequence, $[a_1,a_2,\cdots]^{T}$ to represent a column vector, and $\left\Vert\cdot\right\Vert_p$ to represent the $\ell_p$-norm of a vector.
The symbol $\mathbf{I}_{N}$ indicates the $N\times N$ identity matrix,
$\mathbf{1}_{N}$ indicates the $N\times 1$ vector with all entries being unity,
and $\mathcal{CN}(0,1)$ stands for complex Gaussian distribution with zero mean and unit variance.
The function $[x]^+=\mathrm{max}(x,0)$, while $\mathrm{exp}(\cdot)$ and $\mathrm{ln}(\cdot)$ denote the exponential function and the natural logarithm function, respectively.
The function $\lfloor x\rfloor=\mathrm{max}\{n\in\mathbb{Z}:n\leq x \}$.
Finally, $\mathbb{E}(\cdot)$ means the expectation of a random variable, $\mathbb{I}_{\mathcal{A}}(x)=1$ if $x\in\mathcal{A}$ and zero otherwise, and $\mathcal{O}(\cdot)$ means the order of arithmetic operations.
For ease of retrieval, important variables and parameters to be used in this paper are listed in Table~I.

\begin{table*}[!t]
\caption{Summary of Important Variables and Parameters}
\centering
\begin{tabular}{|l|l|l|}
\hline
\textbf{Symbol} & \textbf{Type} & \textbf{Description} \\
\hline
$p_{k}\in\mathbb{R}_+$  & Variable & Transmit power (in $\mathrm{Watt}$) at user $k$. \\
$v_{m}\in\mathbb{Z}_+$ & Variable & Number of training samples for task $m$.  \\
\hline
$P$ & Parameter & Total transmit power budget (in $\mathrm{Watt}$). \\
$B$ & Parameter & Communication bandwidth (in $\mathrm{Hz}$). \\
$T$ & Parameter & Transmission time (in $\mathrm{s}$). \\
$\xi$ & Parameter & The factor accounting for packet loss and network overhead. \\
$\sigma^2$ & Parameter & Noise power (in $\mathrm{Watt}$). \\
$D_m$ & Parameter & Data size (in $\mathrm{bit}$) per sample for task $m$ \\
$A_m$ & Parameter & Number of historical samples at the edge for task $m$. \\
$G_{k,l}$ & Parameter & The composite channel gain from user $l$ to the edge when decoding data of user $k$. \\
\hline
$\mathcal{D}_k,\mathcal{H}_k,\mathcal{T}_k,\mathcal{V}_k$ & Dataset & The full dataset, historical dataset, training dataset, validation dataset of user $k$. \\
\hline
$\Psi_m(v_m)$ & Function & Classification error of the learning model $m$ when the sample size is $v_m$. \\
$\Theta_m(v_m|a_m,b_m)$ & Function & Empirical classification error model for task $m$ with parameters $(a_m,b_m)$. \\
$U(\mathbf{d})$ & Function & Prediction confidence of sample $\mathbf{d}$. \\
\hline
\end{tabular}
\end{table*}

\section{System Model and Problem Formulation}

\setcounter{secnumdepth}{4}
We consider an edge machine learning system shown in Fig.~1, which consists of an intelligent edge with $N$ antennas and $K$ users with datasets $\{\mathcal{D}_1,\cdots,\mathcal{D}_K\}$.
The goal of the edge is to train $M$ classification models by collecting data from $M$ user groups (e.g., UAVs with camera sensors) $\{\mathcal{Y}_1,\mathcal{Y}_2,\cdots,\mathcal{Y}_M\}$, with the group $\mathcal{Y}_m$ containing all users having data for training the model $m$.
In case where some data from a particular user is used to train both model $m$ and model $j$, we can allow $\mathcal{Y}_m$ and $\mathcal{Y}_j$ to include a common user, i.e., $\mathcal{Y}_m\bigcap\mathcal{Y}_j\neq \emptyset$ for $m\neq j$.
For the classification models, without loss of generality, Fig.~1 depicts a convolutional neural network (CNN) and a support vector machine (SVM) (i.e., $M=2$), but more user groups and other classification models are equally valid.
It is assumed that the data are labeled at the edge.
This can be supplemented by the recent self-labeled techniques \cite{self,self2}, where a classifier is trained with an initial small number of labeled examples, and then the model is retrained with its own most confident predictions, thus enlarging its labeled training set.
After training the classifiers, the edge can feedback the trained models to users for subsequent use (e.g., object detection).
Notice that if the classifiers are pre-trained at the cloud and deployed at the edge, the task of edge machine learning is to fine-tune the pre-trained models at the edge, using local data generated from MTC users.

\begin{figure*}[!t]
 \centering
    \includegraphics[width=170mm]{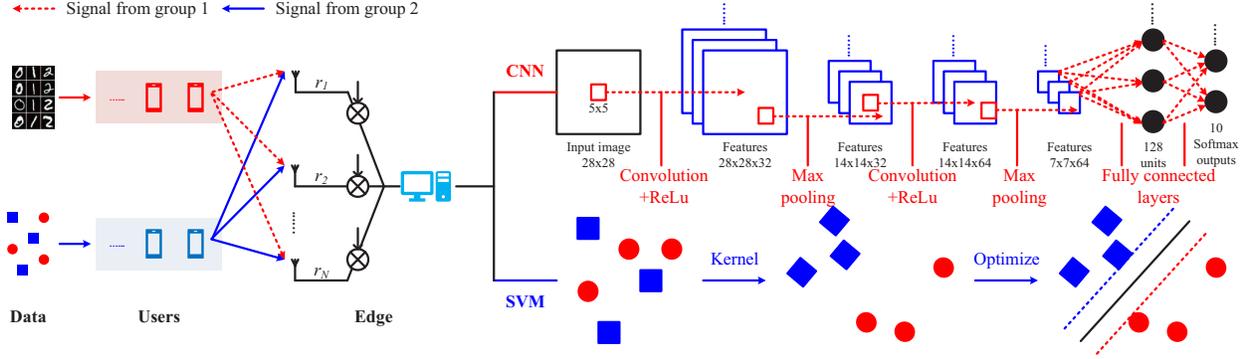}
  \caption{System model of machine intelligence  at the edge.
  }
  \label{fig1}
\end{figure*}

More specifically, the user $k\in\{1,\cdots,K\}$ transmits a signal $s_{k}$ with power $\mathbb{E}[|s_{k}|^2]=p_{k}$.
Accordingly, the received signal $\mathbf{r}=[r_1,\cdots,r_N]^T\in\mathbb{C}^{N\times 1}$ at the edge is
$\mathbf{r}=\sum_{k=1}^K\mathbf{h}_{k}\, s_{k}+\mathbf{n}$,
where $\mathbf{h}_{k}\in \mathbb{C}^{N\times 1}$ is the channel vector from the $k^{\mathrm{th}}$ user to the edge,
and $\mathbf{n}\sim \mathcal{CN}(\mathbf{0},\sigma^2\mathbf{I}_N)$.
By applying the maximal ratio combining (MRC) receiver $\mathbf{h}_{k}/\left\Vert\mathbf{h}_{k}\right\Vert_2$ to $\mathbf{r}$, the data-rate of user $k$ is
\begin{align}
&R_{k}=\mathrm{log}_2\left(1+\frac{G_{k,k}p_{k}}{\sum_{l=1,l\neq k}^KG_{k,l}p_{l}+
\sigma^2} \right), \label{Rk}
\end{align}
where $G_{k,l}$ represents the composite channel gain (including channel fading and MIMO processing) from user $l$ to the edge when decoding data of user $k$:
\begin{align}
&G_{k,l}=
\left\{
\begin{aligned}
&\left\Vert\mathbf{h}_{k}\right\Vert_2^2
,&\mathrm{if}~k=l
\\
&\frac{|\mathbf{h}_k^H\mathbf{h}_{l}|^2}{\left\Vert\mathbf{h}_{k}\right\Vert_2^2}
,&\mathrm{if}~k\neq l
\end{aligned}
\right.
.
\end{align}

With the expression of $R_k$ in \eqref{Rk}, the amount of data in $\mathrm{bit}$ received from user $k$ is $BTR_{k}$, where constant $B$ is the bandwidth in $\mathrm{Hz}$ that is assigned to the system (e.g., a standard MTC system would have $180~\mathrm{kHz}$ bandwidth \cite{iot3}), and $T$ is the total number of transmission time in second.
As a result, the total number of training samples that are collected at the edge for training the model $m$ is
\begin{align}\label{sample size}
&v_m= \xi\sum_{k\in\mathcal{Y}_m} \left\lfloor \frac{BTR_{k}}{D_m} \right\rfloor +A_m \approx \sum_{k\in\mathcal{Y}_m}\frac{\xi BTR_{k}}{D_m} +A_m,
\end{align}
where $A_m$ is the number of historical samples for task $m$ residing at the edge, $D_m$ is the number of bits for each data sample, and $\xi\leq1$ is a factor accounting for the reduced number of samples due to packet loss and network overhead.
The approximation is due to $\lfloor x\rfloor\to x$ when $x\gg 1$.

For example, in real-world edge/cloud machine learning applications, the data needs to be uploaded multiple times \cite{history1,history2} (e.g., twice a day).
If the historical dataset of user $k$ at the edge is $\mathcal{H}_k$, then $A_m=\sum_{k\in\mathcal{Y}_m}|\mathcal{H}_k|$.
For $D_m$, if we consider the MNIST dataset \cite{MNIST}, since the handwritten digit images are gray scale with $28\times 28$ pixels (each pixel has $8$ bits), in this case $D_m=8\times28\times28+4=6276~\mathrm{bits}$ ($4$ bits are reserved for the labels of $10$ classes \cite{MNIST} in case the users also transmit labels).
Lastly, if the system reserves $30\%$ of the resource blocks for network overhead \cite{PDCCH}, and the probability of packet error rate is $0.1$ \cite{xi2}, then $\xi$ can be estimated as $\xi=0.7\times 0.9=0.63$.

In the considered system, the design variables that can be controlled are the transmit powers of different users $\mathbf{p}=[p_1,\cdots,p_K]^T$ and the sample sizes of different models $\mathbf{v}=[v_1,\cdots,v_M]^T$.
Since the power costs at users should not exceed the total budget $P$, the variable $\mathbf{p}$ needs to satisfy
$\sum_{k=1}^Kp_k\leq P$.
Since a larger $\sum_{k=1}^Kp_k$ always improves the learning performance, we can rewrite $\sum_{k=1}^Kp_k\leq P$ as $\sum_{k=1}^Kp_k=P$.
Having the transmit power satisfied, it is then crucial to minimize the classification errors (i.e., the number of incorrect predictions divided by the number of total predictions),
which leads to the following learning centric power allocation (LCPA) problem:
\begin{subequations}
\begin{align}
\mathrm{P}:\mathop{\mathrm{min}}_{\substack{\mathbf{p},\,\mathbf{v}}}
\quad&\mathop{\mathrm{max}}_{m=1,\cdots,M}~\Psi_m(v_m),   \nonumber\\
\mathrm{s. t.}\quad &\sum_{k=1}^Kp_k=P,\quad p_k\geq 0,\quad k=1,\cdots,K,  \\
&\sum_{k\in\mathcal{Y}_m}\frac{\xi BT}{D_m}\mathrm{log}_2\left(1+\frac{G_{k,k}p_{k}}{\sum_{l=1,l\neq k}^KG_{k,l}p_{l}+
\sigma^2} \right)
\nonumber\\
&
+A_m
=v_m,\quad m=1,\cdots,M,
 \label{P0b}
\end{align}
\end{subequations}
where $\Psi_m(v_m)$ is the classification error of the learning model $m$ when the sample size is $v_m$, and the min-max operation at the objective function is to guarantee the worst-case learning performance.
The key challenge to solve $\rm{P}$ is that functions $(\Psi_1,\cdots,\Psi_M)$ represent generalization errors, and to the best of the authors' knowledge, currently there is no exact expression of $\Psi_m(v_m)$.
To address this issue, Section III will adopt an empirical classification error model to approximate $\Psi_m$.

The problem formulation $\mathrm{P}$ has assumed that the data is IID among different users (i.e., different users have identical distributions).
However, practical applications may involve non-IID data distributions.
For example, user $1$ has $1000$ samples but they are exactly the same (e.g., repeating the first sample in the MNIST dataset); and user $2$ has $50$ different samples (e.g., randomly drawn from the MNIST dataset).
In this case, upper bound and lower bound data-rate constraints should be imposed depending on the quality of users' data (i.e., whether their data improves the learning performance and how much the improvement could be):
\begin{itemize}
\item If the data transmitted from user $k$ does not help to improve the learning performance, a maximum amount of transmitted data $Z_k^{\rm{max}}$ should be imposed.
This is the case of user $1$, since the well-understood data should not be transmitted to the edge over and over again.

\item On the other hand, if the data transmitted from user $k$ helps to improve the learning performance, a minimum amount of data $Z_k^{\rm{min}}$ should be imposed on this user.
This is the case of user $2$, since more data from this user would reduce the learning error.

\end{itemize}
Notice that both the lower and upper bounds can be imposed on the same user simultaneously.
Therefore, we can add the following constraint to problem $\mathrm{P}$:
\begin{align}\label{ratebounds}
Z_k^{\rm{min}}&\leq\frac{\xi BT}{D_m}\mathrm{log}_2\left(1+\frac{G_{k,k}p_{k}}{\sum_{l=1,l\neq k}^KG_{k,l}p_{l}+
\sigma^2} \right)
\nonumber\\
&\leq Z_k^{\rm{max}},\quad k=1,\cdots,K.
\end{align}
Of course, the remaining question is how to determine whether the data from a particular user is useful or not.
To this end, the uncertainty sampling method \cite{uncertainty} can be adopted.
In particular, let the function $\mathbb{P}(\mathbf{c}|\mathbf{d}, \mathbf{w})$ denote the probability of the label being $\mathbf{c}$ given sample $\mathbf{d}\in\mathbb{R}^{S\times 1}$ and the model parameter vector $\mathbf{w}$ (e.g., $\mathbf{c}$ is a one-hot vector containing $10$ elements if the learning problem is a $10$-classes classification; for the MNIST dataset, we have $S=784$; and in the considered CNN model, $\mathbf{w}$ contains the convolution matrices and the bias vectors).
The confidence of the predicted label $\mathbf{c}^*=\mathop{\mathrm{argmax}}_{\mathbf{c}}\mathbb{P}(\mathbf{c}|\mathbf{d},\mathbf{w})$ is
\begin{align}\label{uncertain}
&U(\mathbf{d})=\mathbb{P}(\mathbf{c}^*|\mathbf{d},\mathbf{w}).
\end{align}

Therefore, we can compute the $U(\mathbf{d})$ for all the historical samples from user $k$.
If the average (or median, minimum) value of $U(\mathbf{d})$ is large (e.g., $>0.9$) for the historical data from user $k$, then user $k$ should have its data-rate upper bounded.
On the other hand, if the average (or median, minimum) value of $U(\mathbf{d})$ is small (e.g., $<0.5$), then user $k$ should have its data-rate lower bounded.
The entire LCPA scheme with uncertainty sampling is summarized in Fig.~2.

\begin{figure}[!t]
 \centering
    \includegraphics[width=75mm]{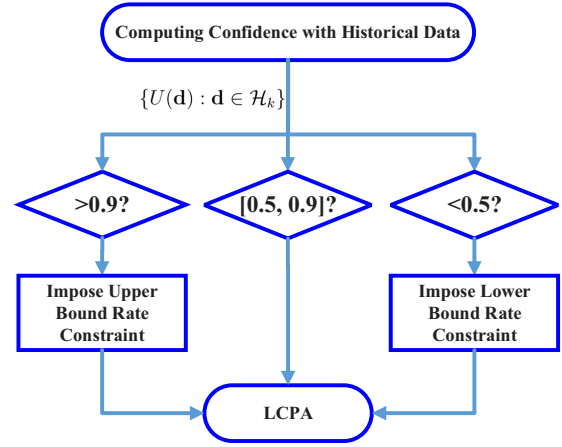}
  \caption{Illustration of the LCPA scheme with uncertainty sampling.
  }
  \label{fig2}
\end{figure}

\emph{Remark 1 (Practical Power Control Procedure):
The LCPA problem will be solved at the edge, which can then inform the users about their transmit powers through downlink control channels.
For example, the 3GPP standard reserves some resource blocks for Physical Downlink Control Channel (PDCCH) \cite{PDCCH}, which are used for sending control signals such as users' transmit powers and modulation orders.
}

\section{Modeling of Classification Error}

\subsection{Classification Error Rate Function}

In general, the classification error $\Psi_m(v_m)$ is a nonlinear function of $v_m$ \cite{model1,model2,model3,model4,model5,model6}.
Particularly, this nonlinear function should satisfy the following properties:
\begin{itemize}
\item[(i)] As $\Psi_m$ is a percentage, it satisfies $0\leq \Psi_m(v_m) \leq 1$;

\item[(ii)] Since more data would provide more information, $\Psi_m(v_m)$ is a monotonically decreasing function of $v_m$ \cite{model1};

\item[(iii)] As $v_m$ increases, the magnitude of the partial derivative $|\partial \Psi_m/\partial v_m|$ would gradually decrease and become zero when $v_m$ is sufficiently large \cite{model2}, meaning that increasing sample size no longer helps machine learning.

\end{itemize}
Based on the properties (i)--(iii), the following nonlinear model $\Theta_m(v_m|a_m,b_m)$ \cite{model1,model2,model5} can be used to capture the shape of $\Psi_m(v_m)$:
\begin{align}
\Psi_m(v_m) \approx
\Theta_m(v_m|a_m,b_m)
=a_m\, v_m^{-b_m}, \label{model1}
\end{align}
where $a_m,b_m\geq 0$ are tuning parameters.
It can be seen that $\Theta_m$ satisfies all the above properties.
Moreover, $\Theta_m(v_m|a_m,b_m)\to 0$ if $v_m\to +\infty$, meaning that the error is $0$ with infinite data.\footnote{We assume the model is powerful enough such that given infinite amount of data, the error rate can be driven to zero.}

\textbf{Interpretation from Learning Theory.}
Apart from (i)--(iii), the model \eqref{model1} corroborates the \emph{inverse power relationship} between learning performance $\Psi_m$ and the amount of training data $v_m$ from the perspective of statistical mechanics \cite{model6}.
In particular, according to \cite{model6}, the training procedure can be modeled as a Gibbs distribution of networks characterized by a temperature parameter $T_g$.
The asymptotic generalization error as the number of samples $v_m$ goes to infinity is expressed as \cite[Eq. (3.12)]{model6}
\begin{align}\label{RQ2_A2}
\Psi_m(v_m)&\rightarrow\epsilon_{\mathrm{min}}+\left(\frac{T_g}{2}+\frac{\mathrm{Tr}(\mathbf{U}_m\mathbf{V}_m^{-1})}{2W_m}\right)W_mv_m^{-1},
\nonumber\\
\mathrm{if}~v_m&\rightarrow +\infty,
\end{align}
where $\epsilon_{\mathrm{min}}\geq 0$ is the minimum error for all possible learning models and $W_m$ is the number of parameters.
The matrices $(\mathbf{U}_m,\mathbf{V}_m)$ contain the second-order and first-order derivatives of the generalization error function with respect to the parameters of model $m$.
By comparing \eqref{RQ2_A2} with \eqref{model1}, we can see that $a_m=\left(\frac{T_g}{2}+\frac{\mathrm{Tr}(\mathbf{U}_m\mathbf{V}_m^{-1})}{2W_m}\right)W_m$ and $b_m=-1$ in the asymptotic case.
Therefore, as the number of samples goes to infinity, the weighting factor $a_m$ in \eqref{model1} accounts for the \emph{model complexity} of the classifier $m$.
Moreover, in the finite sample size regime, the slopes of learning curves may be different for different datasets even for the same machine learning model.
This means that $v_m^{-1}$ is not always suitable in practice.
Therefore, $b_m$ is introduced as a tuning parameter accounting for the correlation in a dataset.

\subsection{Parameter Fitting of CNN and SVM Classifiers}

We use the public MNIST dataset \cite{MNIST} as the input images, and train the $6$-layer CCN (shown in Fig.~1) with training sample size $v_m^{(i)}$ ranging from $100$ to $10000$.
In particular, the input image is sequentially fed into a $5\times 5$ convolution layer (with ReLu activation, 32 channels, and SAME padding), a $2\times 2$ max pooling layer, then another $5\times 5$ convolution layer (with ReLu activation, 64 channels, and SAME padding), a $2\times 2$ max pooling layer, a fully connected layer with $128$ units (with ReLu activation), and a final softmax output layer (with $10$ ouputs).
The training procedure is implemented via Adam optimizer with a learning rate of $10^{-4}$ and a mini-batch size of $100$.
After training for $5000$ iterations, we test the trained model on a dataset with $1000$ unseen samples, and compute the corresponding classification error.
By varying the sample size $v_m$ as $(v_m^{(1)},v_m^{(2)},\cdots)=(100,150,200,300,500,1000,5000,10000)$, we can obtain the classification error $\Psi_m(v_m^{(i)})$ for each sample size $v_m^{(i)}$, where $i=1,\cdots,Q$, and $Q=8$ is the number of points to be fitted.
With $\{v_m^{(i)},\Psi_m(v_m^{(i)})\}_{i=1}^Q$, the parameters $(a_m,b_m)$ in $\Theta_m$ can be found via the following nonlinear least squares fitting:
\begin{align}
\mathop{\mathrm{min}}_{a_m,\,b_m}\quad&\frac{1}{Q}\mathop{\sum}_{i=1}^Q\Big|\Psi_m\left(v_m^{(i)}\right)-\Theta_m\left(v_m^{(i)}|a_m,b_m\right)\Big|^2,
\nonumber\\
\mathrm{s.t.}\quad
&a_m\geq 0,\quad b_m\geq 0.
\label{fitting}
\end{align}
The above problem can be solved by two-dimensional brute-force search, or gradient descent method.
Since the parameters $(a_m,b_m)$ for different tasks are obtained independently, the total complexity is linear in terms of the number of tasks.
The fitted classification error versus the sample size is shown in Fig.~3a.
It is observed from Fig.~3a that with the parameters $(a_m,b_m)=(9.27,0.74)$, the nonlinear classification error model in \eqref{model1} matches the experimental data of CNN very well.

To demonstrate the versatility of the model, we also fit the nonlinear model to the classification error of a support vector machine (SVM) classifier.
The SVM uses penalty coefficient $C=1$ and Gaussian kernel function $K(\mathbf{x}_i,\mathbf{x}_j)=\mathrm{exp}\left(-\widetilde{\gamma}\, \left\Vert\mathbf{x}_i-\mathbf{x}_j\right\Vert_2^2\right)$ with $\widetilde{\gamma}=0.001$ \cite{sklearn}.
Moreover, the SVM classifier is trained on the digits dataset in the Scikit-learn Python machine learning tookbox, and the dataset contains $1797$ images of size $8\times 8$ from $10$ classes, with $5$ bits (corresponding to integers $0$ to $16$) for each pixel \cite{sklearn}.
Therefore, each image needs $D_m=8\times 8\times 5+4=324~\mathrm{bits}$.
Out of all images, we train the SVM using the first $1000$ samples with sample size $(v_m^{(1)},v_m^{(2)},\cdots)=(30,50,100,200,300,400,500,1000)$,
and use the latter $797$ samples for testing.
The parameters $(a_m,b_m)$ for the SVM are obtained following a similar procedure in \eqref{fitting}.
It is observed from Fig.~3a that with $(a_m,b_m)=(6.94,0.8)$, the model in \eqref{model1} fits the experimental data of SVM.

\begin{figure*}
 \centering
  \subfigure[]{
    \includegraphics[width=58mm]{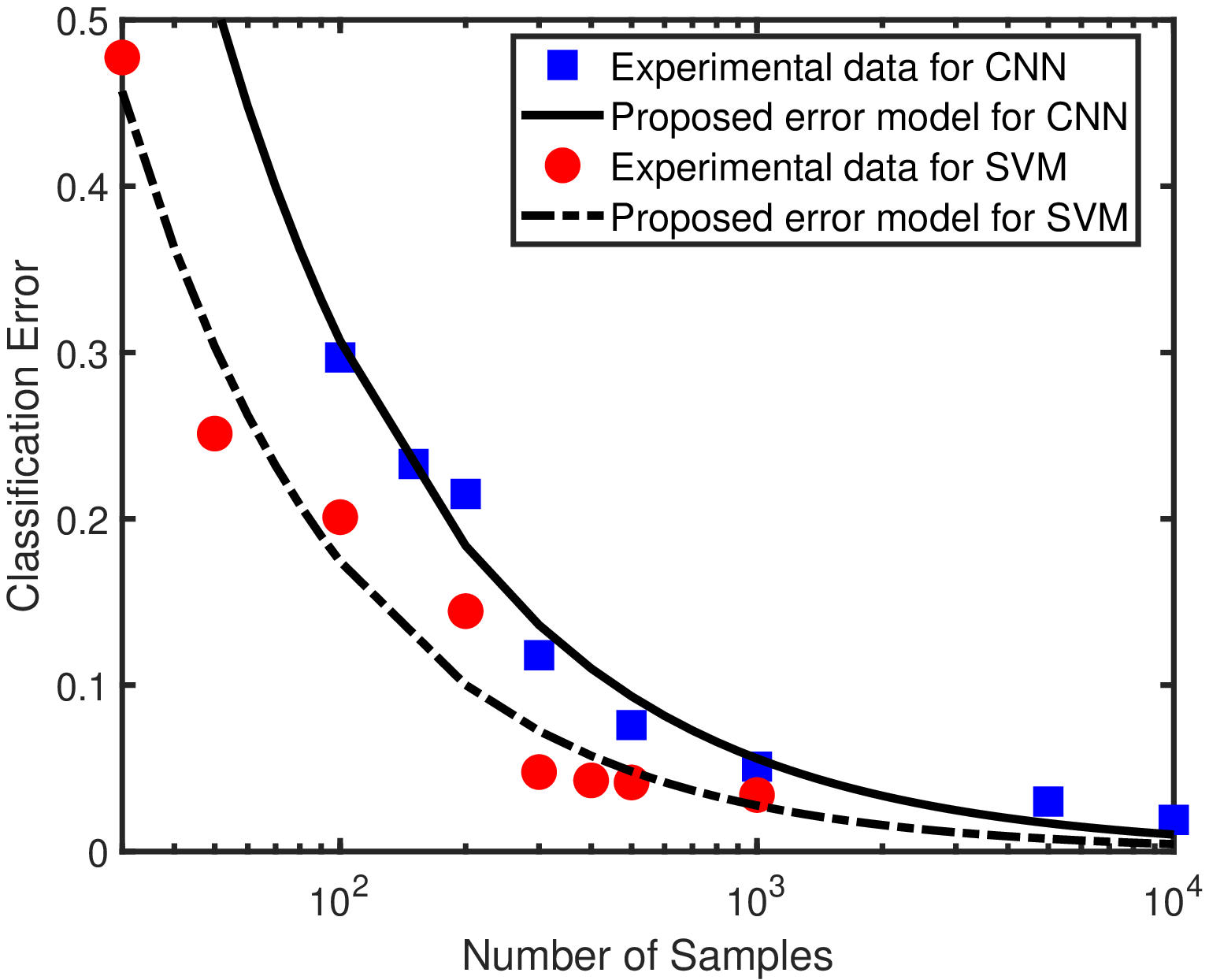}}
  \subfigure[]{
    \includegraphics[width=58mm]{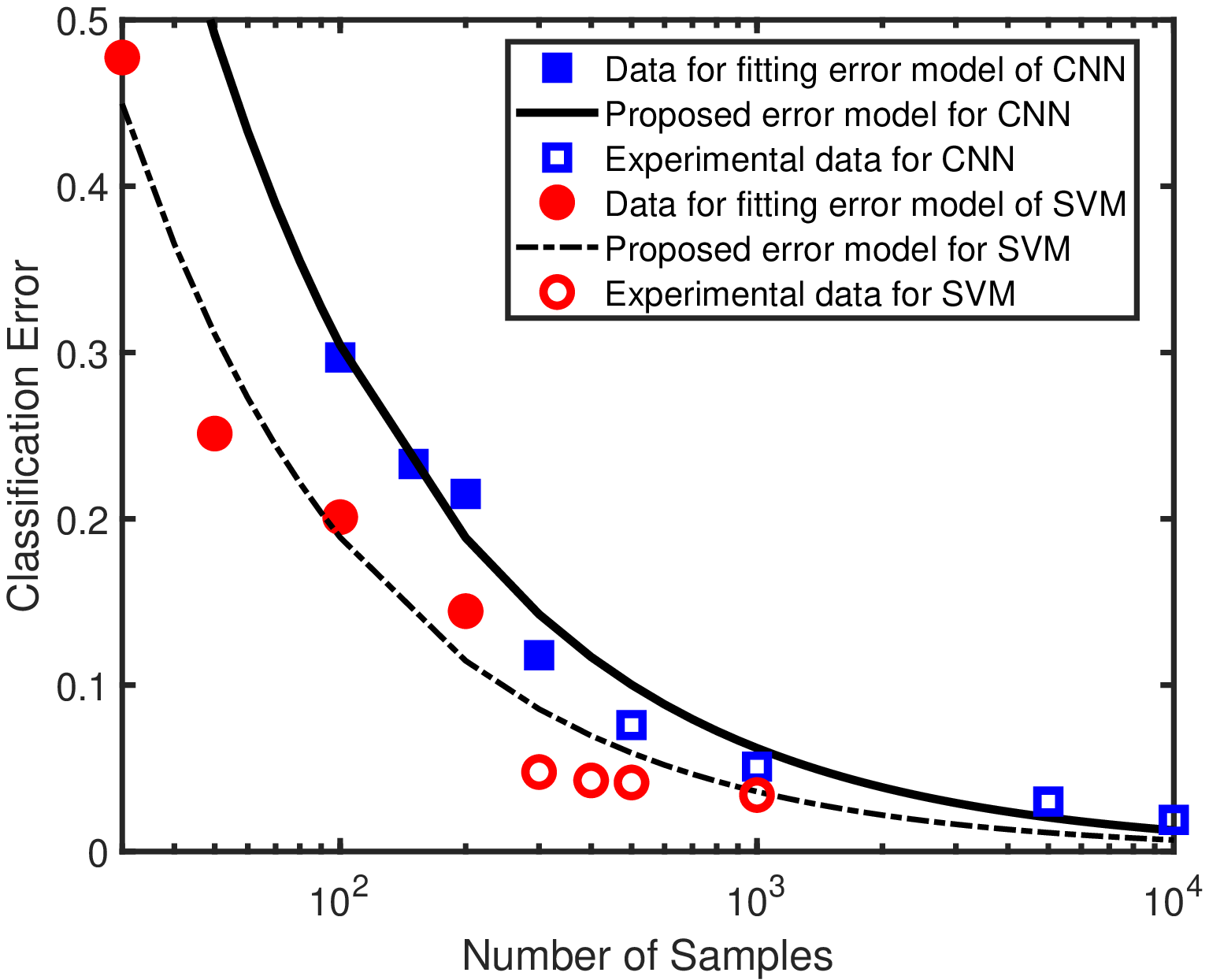}}
  \subfigure[]{
    \includegraphics[width=58mm]{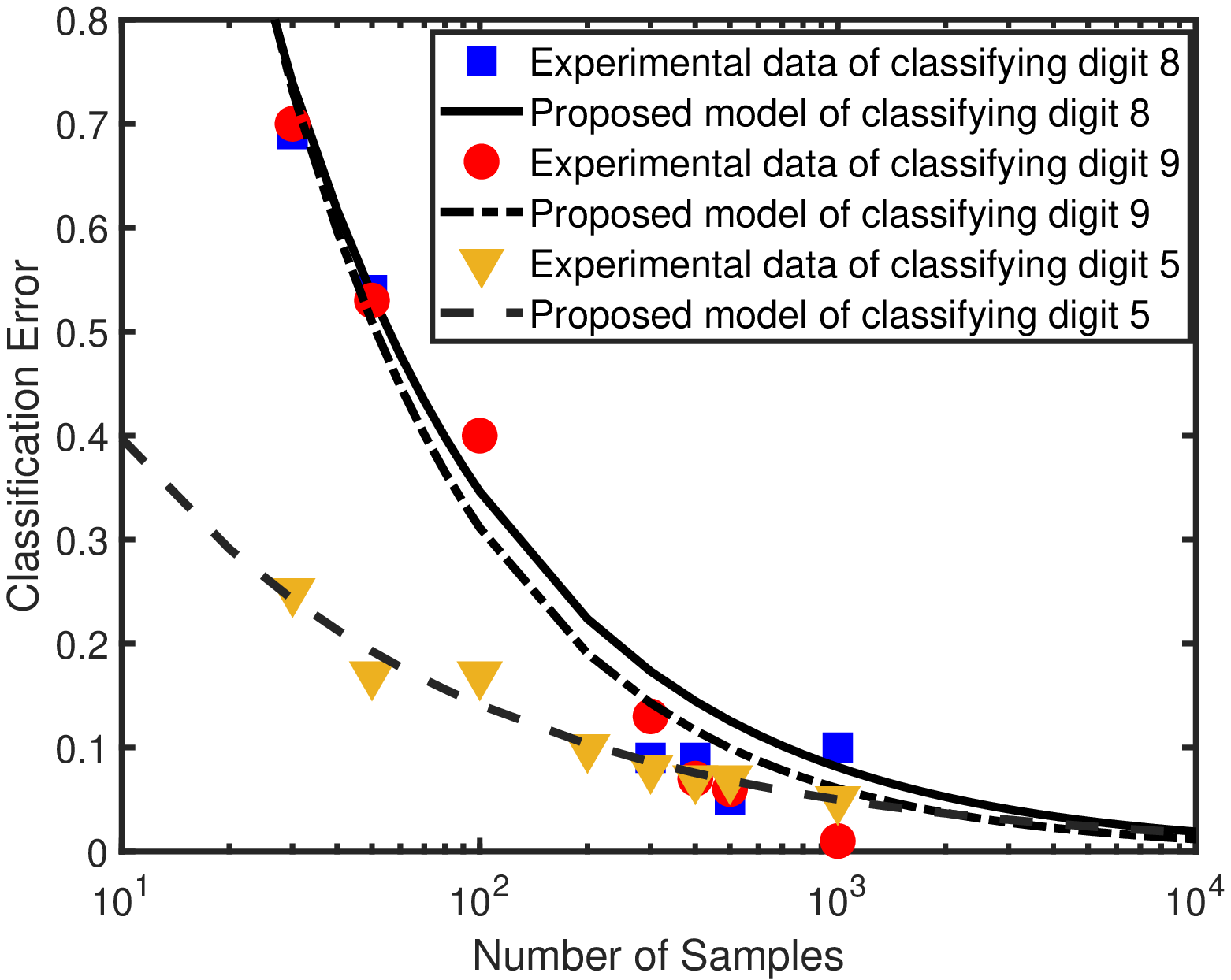}}
  \caption{
  a) Comparison between the experimental data and the nonlinear classification error model. The parameters in the models are given by
$(a_m,b_m)=(9.27,0.74)$ for CNN and $(a_m,b_m)=(6.94,0.8)$ for SVM;
  b) Fitting the error function to historical datasets. The parameters in the models are given by $(a_m,b_m)=(7.3,0.69)$ for CNN and $(a_m,b_m)=(5.2,0.72)$ for SVM;
  c) Comparison between different classification tasks.}
\end{figure*}

\subsection{Practical Implementation}

One may wonder how could one obtain the fitted classification error model before the actual machine learning model is being trained.
There are two ways to address this issue.

\textbf{1) Extrapolation.} More specifically, the error function can be obtained by training the machine learning model on the historical dataset at the edge, and the performance on a future larger dataset can be predicted.
This is called extrapolation \cite{model1}.
For example, by fitting the error function to the first half experimental data of CNN in Fig.~3b (i.e., $v_m=(100,150,200,300)$), we can obtain $(a_m,b_m)=(7.3,0.69)$, and the resultant curve predicts the errors at $v_m=(500,1000,5000,10000)$ very well as shown in Fig.~3b.
Similarly, with $(a_m,b_m)=(5.2,0.72)$ obtained from the experimental data of $v_m=(30,50,100,200)$, the proposed model for SVM matches the classification errors at $v_m=(300,400,500,1000)$.
It can be seen that the fitting performance in Fig.~3b is slightly worse than that in Fig. 3a, as we use smaller number of pilot data.
But since our goal is to distinguish different tasks rather than accurate prediction of the classification errors, the extrapolation method can guide the resource allocation at the edge.

\textbf{2) Approximation.} This means that we can pre-train a large number of commonly-used models offline (not at the edge) and store their corresponding parameters of $(a_m,b_m)$ in a look-up table at the edge.
Then by choosing a set of parameters from the table, the unknown error model at the edge can be approximated.
This is because the error functions can share the same trend for two similar tasks, e.g., classifying digit `$8$' and `$9$' with SVM as shown in Fig.~3c.
Notice that there may be a mismatch between the pre-training task and the real task at the edge.
This is the case between classifying digit `$8$' and `$5$' in Fig.~3c.
As a result, it is necessary to carefully measure the similarity between two tasks when choosing the parameters.

\subsection{Parameter Fitting of ResNet and PointNet}

To verify the nonlinear model in \eqref{model1} under deeper learning models and larger datasets,
we train the $110$-layer deep residual network (ResNet-$110$ with $1.7~$M parameters) \cite{resnet} using the CIFAR-10 dataset as the input images, with training sample size ranging from $5000$ to $50000$.
The image in the CIFAR-10 dataset has $32\times32$ pixels (each pixel has $3~$Bytes representing RGB), and each image sample has a size of $(32\times32\times3+1)\times8=24584~\mathrm{bits}$.
The training procedure is implemented with a diminishing learning rate and a mini-batch size of $100$.
After training for $50000$ iterations ($\sim2.5$ hours), we test the trained model on a dataset with $10000$ unseen samples, and obtain the corresponding classification error.
It can be seen from Fig.~4a that the proposed model with $(a_m,b_m)=(8.15,0.44)$ matches the experimental data of ResNet-$110$ very well.
Moreover, we also consider the PointNet ($3.5~$M parameters), which applies feature transformations and aggregates point features by max pooling \cite[Fig.~2]{pointnet} to classify $3$D point clouds dataset ModelNet40 (see examples in Fig.~4b).
In ModelNet40, there are $12311$ CAD models from $40$ object categories, split into $9843$ for training and $2468$ for testing.
Each sample has $2000$ points with three single-precision floating-point coordinates ($4~$Bytes), and the data size per sample is $(2000\times3\times4+1)\times8=192008~\mathrm{bits}$.
After training for $250$ epochs ($\sim5.5$ hours) with a mini-batch of $32$, the classification error versus the number of samples is obtained in Fig.~4a, and the proposed classification error model with $(a_m,b_m)=(0.96,0.24)$ matches the experimental data of PointNet very well.

\begin{figure*}[!t]
 \centering
          \hspace{0.1in}
      \subfigure[]{
    \includegraphics[width=58mm]{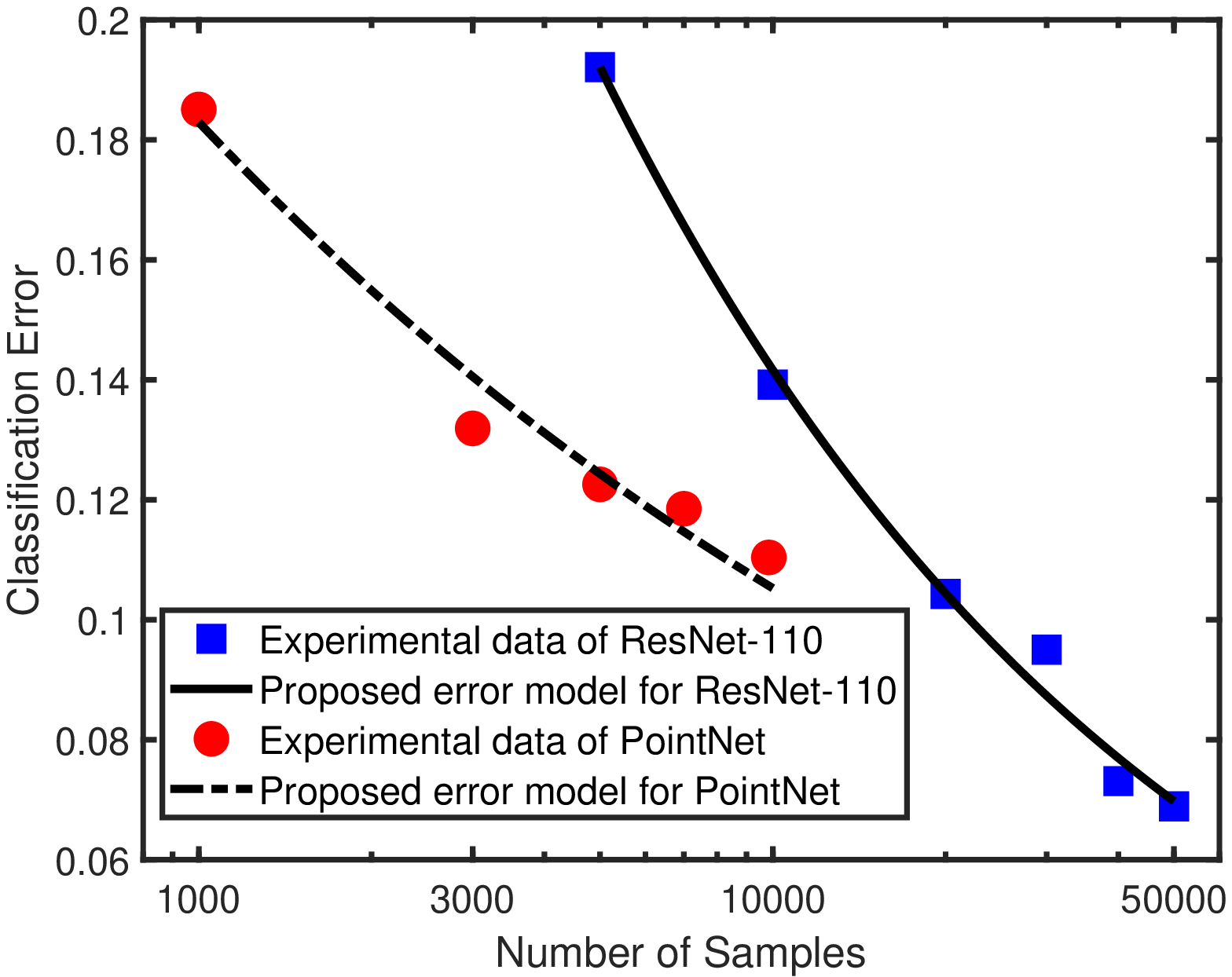}}
          \subfigure[]{
    \includegraphics[width=50mm]{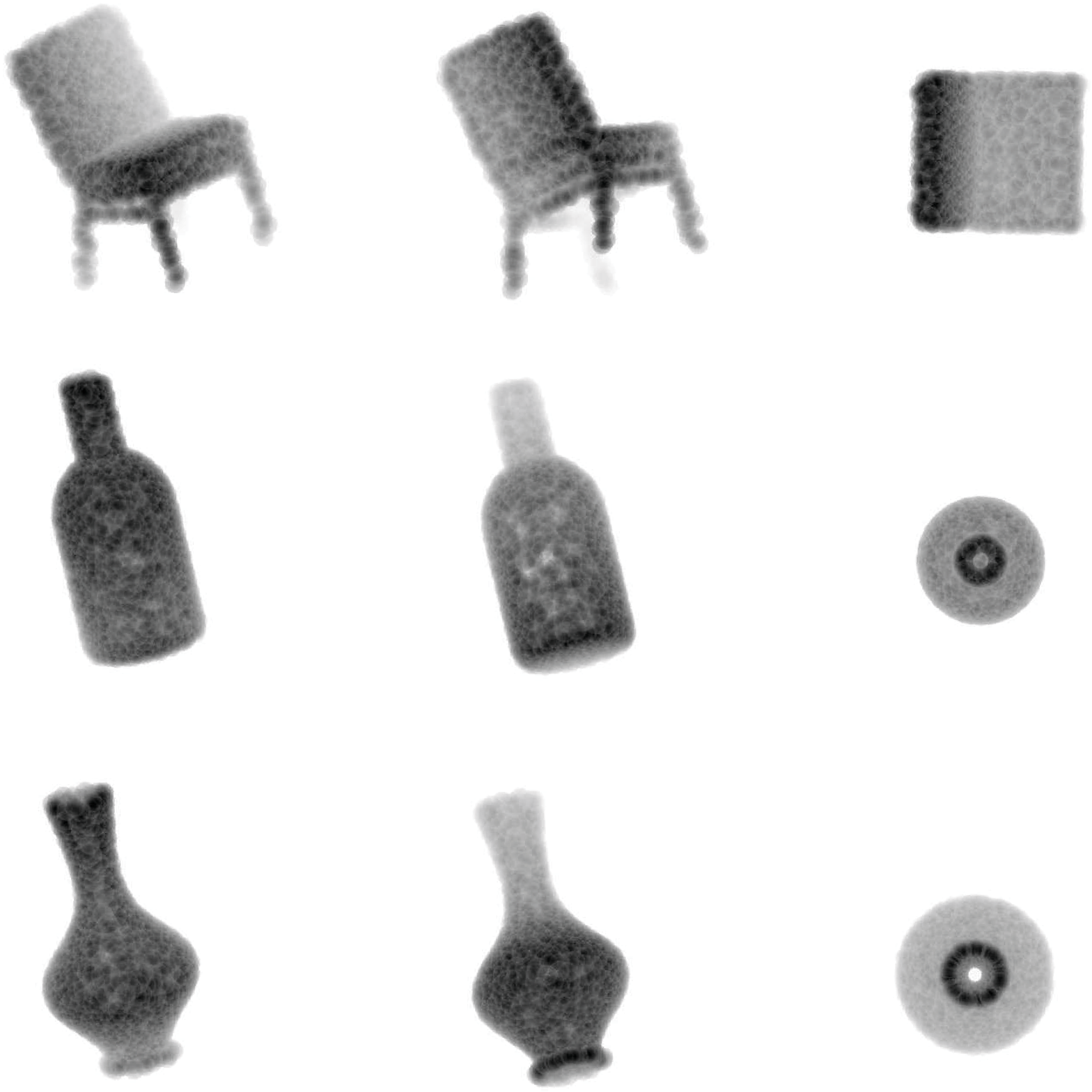}}
              \subfigure[]{
    \includegraphics[width=58mm]{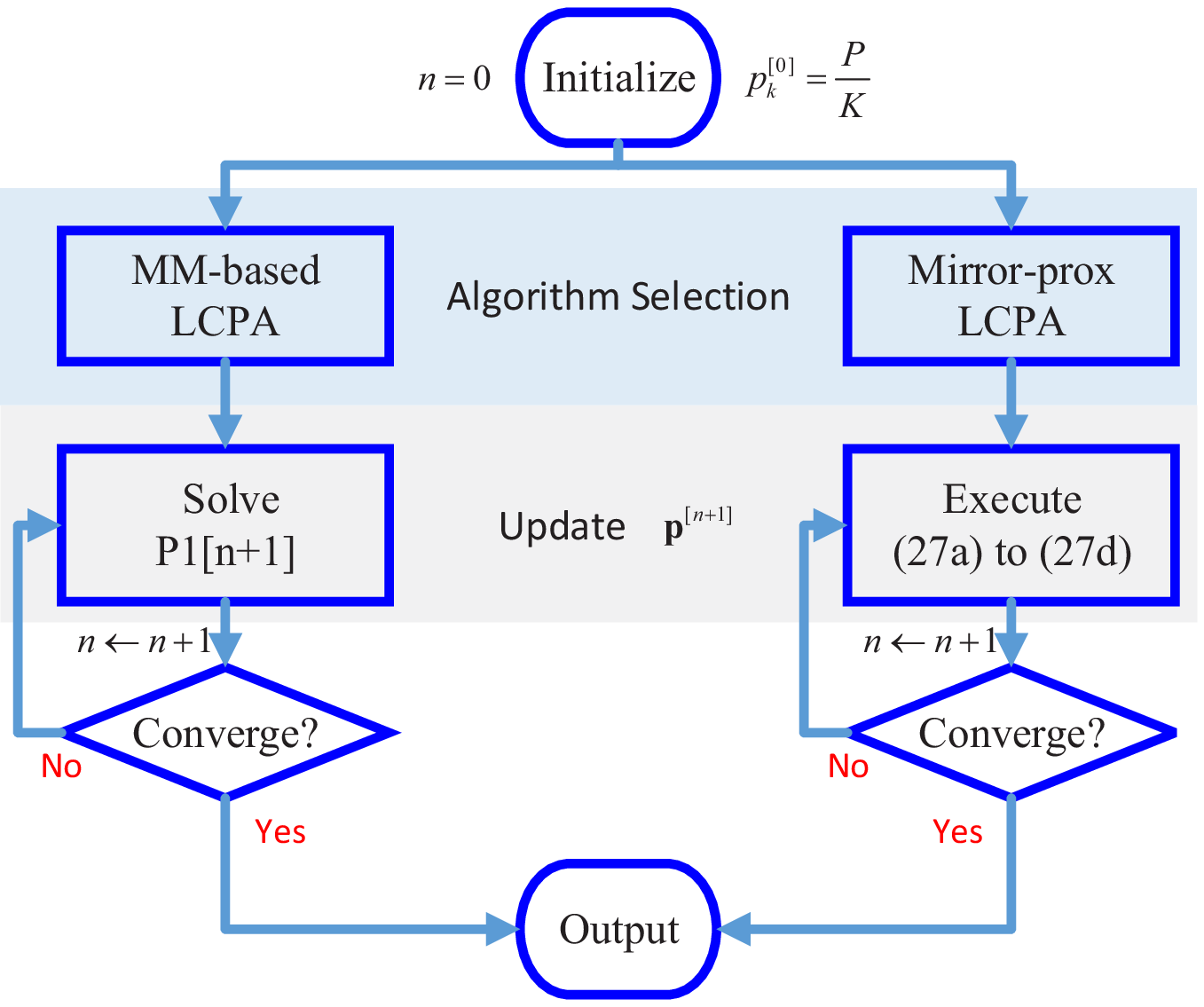}}
  \caption{a) Comparison between the experimental data and the nonlinear classification error model. The parameters in the models are given by
$(a_m,b_m)=(8.15,0.44)$ for ResNet-$110$ and $(a_m,b_m)=(0.96,0.24)$ for PointNet; b) Examples of the 3D point clouds in ModelNet40; c) The flowchart for LCPA Algorithms.}
\end{figure*}

\section{MM-Based LCPA Algorithm}

Based on the results in Section III, we can directly approximate the true error function $\Psi_m$ by $\Theta_m$.
However, to account for the approximation error between $\Psi_m$ and $\Theta_m$ (e.g., due to noise in samples or slight mismatch between data used for fitting and data observed in MTC devices), a weighting factor $\beta_m\geq 1$ can be applied to $\Theta_m$,
where a higher value of $\beta_m$ accounts for a larger approximation error.\footnote{Since the real classification error is scattered around the fitted one, introducing $\rho_m\geq1$ means that we need to elevate the fitted curves to account for the possibilities that worse classification results may happen compared to our prediction.
In other words, we should be more conservative (pessimistic) about our prediction.}
Then by replacing $\Psi_m$ with $\beta_m\Theta_m$ and putting \eqref{P0b} into $\Theta_m(v_m|a_m,b_m)$ to eliminate $\mathbf{v}$, problem $\mathrm{P}$ becomes:
\begin{subequations}
\begin{align}
\mathrm{P}1:\mathop{\mathrm{min}}_{\substack{\mathbf{p}}}
\quad&\mathop{\mathrm{max}}_{m=1,\cdots,M}~\beta_m\,\Phi_m(\mathbf{p}),   \nonumber\\
\mathrm{s. t.}\quad&\sum_{k=1}^Kp_k=P,\quad p_k\geq 0,\quad \forall k,
 \label{P1a} \\
&G_{k,k}p_{k}\geq\left(2^{D_mZ_k^{\rm{min}}/(\xi BT)}-1\right)
\nonumber\\
&
\times\left(\sum_{l\neq k}G_{k,l}p_{l}+\sigma^2\right),\quad \forall k,
 \label{P1b}
\\
&G_{k,k}p_{k}\leq\left(2^{D_mZ_k^{\rm{max}}/(\xi BT)}-1\right)
\nonumber\\
&\times\left(\sum_{l\neq k}G_{k,l}p_{l}+\sigma^2\right),\quad \forall k,
 \label{P1c}
\end{align}
\end{subequations}
where constraints \eqref{P1b}--\eqref{P1c} come from equation \eqref{ratebounds} in the non-IID case and
\begin{align}
&\Phi_m(\mathbf{p}):=
\nonumber\\
&a_m\left[
\sum_{k\in\mathcal{Y}_m} \frac{\xi BT}{D_m}\mathrm{log}_2\left(1+\frac{G_{k,k}p_{k}}{\sum_{l\neq k}G_{k,l}p_{l}+
\sigma^2} \right)
+A_m
\right]^{-b_m}.
\end{align}
It can be seen that $\mathrm{P}1$ is a nonlinear optimization problem due to the nonlinear classification error model \eqref{model1}.
Moreover, the $\mathrm{max}$ operator introduces non-smoothness to the problem, and the objective function is not differentiable.
Thus the existing method based on gradient descent \cite{gradient} is not applicable.

To solve $\mathrm{P}1$, we propose to use the framework of MM \cite{mm1,mm2,mm3,wang}, which constructs a sequence of upper bounds $\{\widetilde{\Phi}_m\}$ on $\{\Phi_m\}$ and replaces $\{\Phi_m\}$ in $\rm{P}1$ with $\{\widetilde{\Phi}_m\}$ to obtain the surrogate problems.
More specifically, given any feasible solution $\mathbf{p}^\star$ to $\mathrm{P}1$, we define surrogate functions
\begin{align}
&\widetilde{\Phi}_{m}(\mathbf{p}|\mathbf{p}^\star)
\nonumber\\
&=a_m\Bigg\{
\sum_{k\in\mathcal{Y}_m} \frac{\xi BT}{D_m\mathrm{ln}2\,}
\Bigg[
\mathrm{ln}\left(\sum_{l=1}^K\frac{G_{k,l}p_{l}}{\sigma^2}+1\right)
\nonumber\\
&\quad
{}
-
\mathrm{ln}\left(\sum_{l=1,l\neq k}^K\frac{G_{k,l}p^\star_{l}}{\sigma^2}+1\right)-\left(\sum_{l=1,l\neq k}^K\frac{G_{k,l}p^\star_{l}}{\sigma^2}+1\right)^{-1}
\nonumber
\\&\quad
{}
\times\left(\sum_{l=1,l\neq k}^K\frac{G_{k,l}p_{l}}{\sigma^2}+1\right)
+1
\Bigg]+A_m
\Bigg\}^{-b_m}, \label{Phi}
\end{align}
and the following proposition can be established.
\begin{proposition}
The functions $\{\widetilde{\Phi}_m\}$ satisfy the following conditions:

\noindent(i) Upper bound condition: $\widetilde{\Phi}_{m}(\mathbf{p}|\mathbf{p}^\star)\geq \Phi_{m}(\mathbf{p})$.

\noindent(ii) Convexity: $\widetilde{\Phi}_m(\mathbf{p}|\mathbf{p}^{\star})$ is convex in $\bm{\mathbf{p}}$.

\noindent(iii) Local equality condition: $\widetilde{\Phi}_{m}(\mathbf{p}^\star|\mathbf{p}^\star)=\Phi_{m}(\mathbf{p}^\star)$ and $\nabla_{\mathbf{p}}\widetilde{\Phi}_{m}(\mathbf{p}^\star|\mathbf{p}^\star)=\nabla_{\mathbf{p}}\Phi_{m}(\mathbf{p}^\star)$.

\end{proposition}
\begin{proof}
See Appendix A.
\end{proof}

With part (i) of \textbf{Proposition 1}, an upper bound can be directly obtained if we replace the functions $\{\Phi_m\}$ by $\{\widetilde{\Phi}_m\}$ around a feasible point.
However, a tighter upper bound can be achieved if we treat the obtained solution as another feasible point and continue to construct the next-round surrogate function.
In particular, assuming that the solution at the $n^{\mathrm{th}}$ iteration is given by $\mathbf{p}^{[n]}$, the following problem is considered at the $(n+1)^{\mathrm{th}}$ iteration:
\begin{align}
\mathrm{P}1[n+1]:\mathop{\mathrm{min}}_{\substack{\mathbf{p}}}
\quad&\mathop{\mathrm{max}}_{m=1,\cdots,M}~\beta_m\widetilde{\Phi}_m(\mathbf{p}|\mathbf{p}^{[n]}),
\nonumber\\
\mathrm{s. t.}\quad&\rm{constraints}~\eqref{P1a}-\eqref{P1c}. \label{P1[n+1]}
\end{align}

Based on part (ii) of \textbf{Proposition 1}, the problem $\mathrm{P}1[n+1]$ is convex and can be solved by off-the-shelf software packages (e.g., CVX Mosek \cite{opt1}) for convex programming.
Denoting its optimal solution as
$\mathbf{p}^*$, we can set
$\mathbf{p}^{[n+1]}=\mathbf{p}^*$, and the process repeats with solving the problem $\mathrm{P}1[n+2]$.
According to part (iii) of \textbf{Proposition 1} and \cite[Theorem 1]{mm1}, every limit point of the sequence
$(\mathbf{p}^{[0]},\mathbf{p}^{[1]},\cdots)$ is a KKT solution to $\mathrm{P}1$ as long as the starting point $\mathbf{p}^{[0]}$ is feasible to $\mathrm{P}1$ (e.g., $\mathbf{p}^{[0]}=P/K\,\bm{1}_K$).
The entire procedure of the MM-based LCPA is summarized in the left hand branch of Fig.~4c.

In terms of computational complexity, $\mathrm{P}1[n+1]$ involves $K$ primal variables and $M+3K+1$ dual variables.
The dual variables correspond to $M+3K+1$ constraints in $\mathrm{P}1[n+1]$, where $M$ constraints come from the $\mathrm{max}$ operator, $K$ constraints come from nonnegative power constraints, $2K$ constraints come from the data-rate bounds, and one constraint comes from the power budget.
Therefore, the worst-case complexity for solving $\mathrm{P}1[n+1]$ is
$\mathcal{O}\Big((M+4K+1)^{3.5}\Big)$ \cite{opt2}.
Consequently, the total complexity for solving $\mathrm{P}1$ is $\mathcal{O}\Big(\mathcal{I}\,(M+4K+1)^{3.5}\Big)$, where $\mathcal{I}$ is the number of iterations needed for the algorithm to converge.

\section{Asymptotic Analysis and Insights to LCPA}

To understand how LCPA works, this section investigates the asymptotic case when the number of antennas at the edge approaches infinity (i.e., $N \to +\infty$).
Moreover, we consider the special case of $|\mathcal{Y}_m|=1$ and $\mathcal{Y}_m\bigcap\mathcal{Y}_j=\emptyset$ (i.e., each user group has only one unique user).
For notational simplicity, we denote the unique user in group $\mathcal{Y}_m$ as user $m=k$.
As each task only involves one user, non-IID data distribution among users in a single task does not exist, and therefore constraints \eqref{P1b}--\eqref{P1c} can be removed.
On the other hand, as $N\to+\infty$, the channels from different users to the edge would be asymptotically orthogonal \cite{massive1,massive2,massive3} and we have
\begin{align}
&G_{k,l}=\frac{|\mathbf{h}_k^H\mathbf{h}_l|^2}{\left\Vert\mathbf{h}_k\right\Vert_2^2}\to 0,\quad\forall k\neq l.
\end{align}
Based on such orthogonality feature, and putting $G_{k,l}=0$ for $k\neq l$ into $\Phi_m$ in $\rm{P}1$, problem $\mathrm{P}1$ under $N\to+\infty$ and $|\mathcal{Y}_m|=1$ is rewritten as
\begin{align}
&\mathrm{P}2:\mathop{\mathrm{min}}_{\substack{\mathbf{p},\,\mu}}
\quad\mu,
\nonumber\\
&\mathrm{s. t.}\quad
\beta_ka_k\left(
\frac{\xi BT}{D_k}\mathrm{log}_2\left(1+\frac{G_{k,k}p_{k}}{\sigma^2} \right)+A_k
\right)^{-b_k}
\leq\mu,\quad \forall k,
\nonumber\\
&\quad\quad
\sum_{k=1}^Kp_k= P,\quad p_k\geq 0,\quad\forall k,
\label{P2b}
\end{align}
where $\mu\in[0,1]$ is a slack variable and has the interpretation of \emph{classification error level}.
The following proposition gives the optimal solution to $\mathrm{P}2$.
\begin{proposition}
The optimal $\mathbf{p}^*$ to $\mathrm{P}2$ is
\begin{align}\label{pk*}
p_k^*(\mu)&=\Bigg[\frac{\sigma^2}{G_{k,k}}\,\mathrm{exp}\left(\frac{D_k\mathrm{ln}2\,}{\xi BT}\left[\left(\frac{\mu}{\beta_ka_k}\right)^{-1/b_k}-
A_k\right]\right)
\nonumber\\
&\quad{}
-\frac{\sigma^2}{G_{k,k}}\Bigg]^+,\quad k=1,\cdots,K,
\end{align}
where $\mu$ satisfies $\sum_{k=1}^Kp_k^*(\mu)=P$.
\end{proposition}
\begin{proof}
See Appendix B.
\end{proof}

To efficiently compute the classification error level $\mu$, it is observed that the function $p_k^*(\mu)$ is a decreasing function of $\mu$.
Therefore, the classification error level $\mu$ can be obtained from solving $\sum_{k=1}^Kp_k^*(\mu)=P$ using bisection method within interval $[0,1]$.
More specifically, given $\mu_{\mathrm{max}}$ and $\mu_{\mathrm{min}}$ (initially $\mu_{\mathrm{max}}=1$ and $\mu_{\mathrm{min}}=0$), we set $\mu=(\mu_{\mathrm{max}}+\mu_{\mathrm{min}})/2$.
If $\sum_kp_k^*(\mu)\geq P$, we update $\mu_{\mathrm{min}}=\mu$; otherwise, we update $\mu_{\mathrm{max}}=\mu$. This procedure is repeated until $|\mu_{\mathrm{max}}-\mu_{\mathrm{min}}|<\epsilon$ with $\epsilon=10^{-8}$.
Since bisection method has a linear convergence rate \cite{bisection}, and in each iteration we need to compute $K$ scalar functions $p_1^*(\mu),\cdots,p_K^*(\mu)$, the bisection method has a complexity of $\mathcal{O}(\mathrm{log}\left(1/\epsilon\right)K)$.

\textbf{Scaling Law of Learning Centric Communication.}
According to \textbf{Proposition 2}, the user transmit power $p_k$ is inversely proportional to the wireless channel gain $G_{k,k}=\left\Vert\mathbf{h}_k\right\Vert_2^2$.
On the other hand, it is exponentially dependent on the classification error level $\mu$ and the learning parameters $(a_k,b_k,D_k,A_k)$.
Moreover, among all parameters, $b_k$ is the most important factor, since $b_k$ is involved in both the power and exponential functions.
The above observations disclose that in edge machine learning systems, the learning parameters will have more significant impacts on the radio resource allocation than those of the wireless channels.

\textbf{Learning Centric versus Communication Centric Power Allocation.} Notice that the result in \eqref{pk*} is fundamentally different from the most well-known resource allocation schemes (e.g., iterative water-filling \cite{waterfilling} and max-min fairness \cite{fair}).
For example, the water-filling solution for maximizing the system throughput under $N\to+\infty$ is given by
\begin{align}\label{water}
&p_k^{\rm{WF}}=
\left(\frac{1}{\lambda\mathrm{ln}2\,}-\frac{\sigma^2}{G_{k,k}}\right)^+,
\end{align}
where $\lambda$ is a constant chosen such that $\sum_{k=1}^Kp_k^{\rm{WF}}=P$.
On the other hand, the max-min fairness solution under $N\to+\infty$ is given by
\begin{align}\label{fair}
&p_k^{\rm{FAIR}}=
P\left(\sum_{k=1}^K\frac{\sigma^2}{G_{k,k}} \right)^{-1}\,\frac{\sigma^2}{G_{k,k}}.
\end{align}
It can be seen from \eqref{water} and \eqref{fair} that the water-filling scheme would allocate more power resources to better channels, and the max-min fairness scheme would allocate more power resources to worse channels.
But no matter which scheme we adopt, the only impact factor is the channel condition $\sigma^2/G_{k,k}$.

\section{Large-scale Optimization under IID Datasets}

Although a KKT solution to $\mathrm{P}1$ has been derived in Section IV, it can be seen that MM-based LCPA requires a cubic complexity with respect to $K$.
This leads to time-consuming computations if $K$ is in the range of hundreds or more.
As a result, low-complexity large-scale optimization algorithms are indispensable.
To this end, in this section we consider the case of $N \to +\infty$ under IID datasets, and develop an algorithm based on the FOM.

As $N\to+\infty$, we put $G_{k,l}=0$ for $k\neq l$ into $\Phi_m$ in $\rm{P}1$, and the function $\Phi_m$ is asymptotically equal to
\begin{align}
&\Xi_m(\mathbf{p})
=a_m\left[
\sum_{k\in\mathcal{Y}_m} \frac{\xi BT}{D_m}\mathrm{log}_2\left(1+\frac{G_{k,k}p_k}{\sigma^2} \right)+A_m
\right]^{-b_m}. \label{Xi}
\end{align}
Therefore, the problem $\mathrm{P}1$ in the case of IID datasets and $N \to +\infty$ is equivalent to
\begin{align}
\mathrm{P}3:\mathop{\mathrm{min}}_{\substack{\mathbf{p}}}
\quad&\mathop{\mathrm{max}}_{m=1,\cdots,M}~\beta_m\Xi_m(\mathbf{p}),
\nonumber\\
\mathrm{s. t.}\quad&\sum_{k=1}^Kp_k=P,\quad p_k\geq 0,\quad k=1,\cdots,K. \label{P3}
\end{align}

The major challenge for solving $\mathrm{P}3$ comes from the nonsmooth operator $\rm{max}$ in the objective function, which hinders us from computing the gradients.
To deal with the non-smoothness, we reformulate $\mathrm{P}3$ into a smooth bilevel optimization problem with $\ell_1$-norm (simplex) constraints.
Observing that the projection onto a simplex in Euclidean space requires high computational complexities, a mirror-prox LCPA method working on non-Euclidean manifold is proposed.
In this way, the distance is measured by Kullback-Leibler (KL) divergence, and the non-Euclidean projection would have analytical expressions.
Lastly, with an extragradient step, the proposed mirror-prox LCPA converges to the global optimal solution to $\mathrm{P}3$ with an iteration complexity of $\mathcal{O}(1/\epsilon)$ \cite{BFOM1,BFOM2,BFOM3}, where $\epsilon$ is the target solution accuracy.

More specifically, we first equivalently transform $\mathrm{P}3$ into a smooth bilevel optimization problem.
By defining set $\mathcal{P}=\left\{\mathbf{p}\in\mathbb{R}^{K\times 1}_+:\left\Vert\mathbf{p}\right\Vert_1=P\right\}$ and introducing variables $\bm{\alpha}\in\mathbb{R}^{M\times 1}$ such that $\bm{\alpha}\in\mathcal{A}=\{\bm{\alpha}\in\mathbb{R}^{M\times 1}_+:\left\Vert\bm{\alpha}\right\Vert_1=1\}$, $\mathrm{P}3$ is rewritten as
\begin{align}\label{saddle}
&\mathrm{P}4:\mathop{\mathrm{min}}_{\substack{\mathbf{p}\in\mathcal{P}}}
~\mathop{\mathrm{max}}_{\bm{\alpha}\in\mathcal{A}}~\underbrace{\mathop{\sum}_{m=1}^{M}\alpha_m\times\beta_m\Xi_m(\mathbf{p})}_{\Upsilon(\bm{\alpha},\mathbf{p})}.
\end{align}
It can be seen from $\mathrm{P}4$ that $\Upsilon(\bm{\alpha},\mathbf{p})$ is differentiable with respect to either $\mathbf{p}$ or $\bm{\alpha}$, and the corresponding gradients are
\begin{subequations}
\begin{align}
\nabla_{\mathbf{p}}\Upsilon(\bm{\alpha},\mathbf{p})&=
\mathop{\sum}_{m=1}^{M}\alpha_m\beta_m\nabla_\mathbf{p}\Xi_m(\mathbf{p}), \label{gradient1}
\\
\nabla_{\bm{\alpha}}\Upsilon(\bm{\alpha},\mathbf{p})
&=
\left[\beta_1\Xi_1(\mathbf{p}),\cdots,\beta_M\Xi_M(\mathbf{p}) \right]^T, \label{gradient2}
\end{align}
\end{subequations}
where
\begin{align}
\nabla_\mathbf{p}\Xi_m(\mathbf{p})
&=
\left[\frac{\partial \Xi_m}{\partial p_1},\cdots,\frac{\partial \Xi_m}{\partial p_K} \right]^T, \label{gradient}
\end{align}
with its $j^{\mathrm{th}}$ element being
\begin{align}
\frac{\partial \Xi_m}{\partial p_j}&=
-\frac{a_mb_m\xi BT\,\mathbb{I}_{\mathcal{Y}_m}(j)}{D_m\mathrm{ln}2\,(\sigma^2G_{j,j}^{-1}+p_j)}
\nonumber\\
&\quad{}
\times
\left[
\sum_{k\in\mathcal{Y}_m} \frac{\xi BT}{D_m}\mathrm{log}_2\left(1+\frac{G_{k,k}p_{k}}{\sigma^2} \right)+A_m
\right]^{-b_m-1}.
\label{ximpartial}
\end{align}

However, $\mathrm{P}4$ is a bilevel problem, with both the upper layer variable $\mathbf{p}$ and the lower layer variable $\bm{\alpha}$ involved in the simplex constraints.
In order to facilitate the projection onto simplex constraints, below we consider a non-Euclidean (Banach) space induced by $\ell_1$-norm.
In such a space, the Bregman distance between two vectors $\mathbf{x}$ and $\mathbf{y}$ is the KL divergence
\begin{align}\label{W}
&\mathbb{KL}\left(\mathbf{x},\mathbf{y}\right)=\mathop{\sum}_{l=1,2,\cdots}x_l\,\mathrm{ln}\left(\frac{x_l}{y_l}\right),
\end{align}
and the following proposition can be established.
\begin{proposition}
If the classification error $\mu=\mathop{\mathrm{max}}_{m}\beta_m\Xi_m(\mathbf{p})$ is upper bounded by $\mu\leq\mu_0$, then $\Upsilon(\bm{\alpha},\mathbf{p})$ is $(L_{1},L_{2},L_{2},0)$--smooth in Banach space induced by $\ell_1$-norm, where
\begin{subequations}
\begin{align}
L_{1}&=\mathop{\mathrm{max}}_{\substack{m=1,\cdots,M\\k=1,\cdots,K}}~
\frac{\beta_ma_mb_m\xi BTG_{k,k}}{D_m\mathrm{ln}2\,\sigma^4}
\left(\frac{\mu_0}{\beta_ma_m}\right)^{1+1/b_m}
\nonumber\\
&\quad\times
\left[G_{k,k}+
\frac{(b_m+1)\xi BTH_m}{D_m\mathrm{ln}2}
\left(\frac{\mu_0}{\beta_ma_m}\right)^{1/b_m}\right], \label{L11}
\\
L_{2}&=\mathop{\mathrm{max}}_{m=1,\cdots,M}~
\frac{\beta_ma_mb_m\xi BTH_m}{D_m\mathrm{ln}2\,\sigma^2}
\left(\frac{\mu_0}{\beta_ma_m}\right)^{1+1/b_m}, \label{L12}
\end{align}
\end{subequations}
with $H_m:=\left\Vert[\mathbb{I}_{\mathcal{Y}_m}(1)\,G_{1,1},\cdots,\mathbb{I}_{\mathcal{Y}_m}(K)\,G_{K,K}]^T\right\Vert_2$.
\end{proposition}
\begin{proof}
See Appendix C.
\end{proof}

The smoothness result in \textbf{Proposition 3} enables us to apply mirror descent (i.e., generalized gradient descent in non-Euclidean space) to $\mathbf{p}$ and mirror ascent to $\bm{\alpha}$ in the $\ell_1$-space \cite{BFOM2}.
This leads to the proposed mirror-prox LCPA, which is an iterative algorithm that involves i) a proximal step and ii) an extragradient step.
In particular, the mirror-prox LCPA initially chooses a feasible $\mathbf{p}=\mathbf{p}^{[0]}\in\mathcal{P}$ and $\bm{\alpha}=\bm{\alpha}^{[0]}\in\mathcal{A}$ (e.g., $\mathbf{p}^{[0]}=P/K\,\bm{1}_K$ and $\bm{\alpha}^{[0]}=1/M\,\bm{1}_M$).
Denoting the solution at the $n^{\mathrm{th}}$ iteration as $(\mathbf{p}^{[n]},\bm{\alpha}^{[n]})$, the following equations are used to update the next-round solution \cite{BFOM2}:
\begin{subequations}
\begin{align}
\mathbf{p}^{\diamond}&=
\mathop{\mathrm{argmin}}_{\mathbf{p}\in\mathcal{P}}~
\mathbb{KL}\left(\mathbf{p},\mathbf{p}^{[n]}\right)
\nonumber\\
&\quad{}
+\eta \, \mathbf{p}^T
\left[\mathop{\sum}_{m=1}^{M}\alpha_m^{[n]}\beta_m\nabla_\mathbf{p}\Xi_m(\mathbf{p}^{[n]})\right],
\label{BFOM1}
\\
\bm{\alpha}^{\diamond}&=
\mathop{\mathrm{argmin}}_{\bm{\alpha}\in\mathcal{A}}~
\mathbb{KL}\left(\bm{\alpha},\bm{\alpha}^{[n]}\right)
\nonumber\\
&\quad{}
-\eta
\left[\beta_1\Xi_1(\mathbf{p}^{[n]}),\cdots,\beta_M\Xi_M(\mathbf{p}^{[n]}) \right]\bm{\alpha}, \label{BFOM2}
\\
\mathbf{p}^{[n+1]}&=
\mathop{\mathrm{argmin}}_{\mathbf{p}\in\mathcal{P}}~
\mathbb{KL}\left(\mathbf{p},\mathbf{p}^{[n]}\right)
\nonumber\\
&\quad{}
+\eta \, \mathbf{p}^T
\left[\mathop{\sum}_{m=1}^{M}\alpha_m^{\diamond}\beta_m\nabla_\mathbf{p}\Xi_m(\mathbf{p}^{\diamond})\right],
\label{BFOM3}
\\
\bm{\alpha}^{[n+1]}&=
\mathop{\mathrm{argmin}}_{\bm{\alpha}\in\mathcal{A}}~\mathbb{KL}\left(\bm{\alpha},\bm{\alpha}^{[n]}\right)
\nonumber\\
&\quad{}
-\eta
\left[\beta_1\Xi_1(\mathbf{p}^{\diamond}),\cdots,\beta_M\Xi_M(\mathbf{p}^{\diamond}) \right]\bm{\alpha}, \label{BFOM4}
\end{align}
\end{subequations}
where $\eta$ is the step-size, and the terms inside $[\cdots]$ in \eqref{BFOM1}--\eqref{BFOM4} are obtained from \eqref{gradient1}--\eqref{gradient2}.
Notice that a small $\eta$ would lead to slow convergence of the algorithm while a large $\eta$ would cause the algorithm to diverge.
According to \cite{BFOM2}, $\eta$ should be chosen inversely proportional to Lipschitz constant $L_1$ or $L_2$ derived in \textbf{Proposition 3}.
In this paper, we set $\eta=10^3/L_2$ with $\mu_0=0.1$, which empirically provides fast convergence of the algorithm.

\textbf{How Mirror-prox LCPA Works.}
The formulas \eqref{BFOM1}--\eqref{BFOM2} update the variables along their gradient direction, while keeping the updated point $\{\mathbf{p}^{\diamond},\bm{\alpha}^{\diamond}\}$ close to the current point $\{\mathbf{p}^{[n]},\bm{\alpha}^{[n]}\}$.
This is achieved via the \emph{proximal operator} that minimizes the distance $\mathbb{KL}\left(\mathbf{p},\mathbf{p}^{[n]}\right)$ (or $\mathbb{KL}\left(\bm{\alpha},\bm{\alpha}^{[n]}\right)$) plus a first-order linear function.
Since the KL divergence is the  Bregman distance, the update \eqref{BFOM1}--\eqref{BFOM2} is a Bregman proximal step.
On the other hand, the gradients in \eqref{BFOM3} and \eqref{BFOM4} are computed using $\mathbf{p}^{\diamond}$ and $\bm{\alpha}^{\diamond}$, respectively.
By doing so, we can obtain the look-ahead gradient at the intermediate point $\mathbf{p}^{\diamond}$ and $\bm{\alpha}^{\diamond}$ for updating $\mathbf{p}^{[n+1]}$ and $\bm{\alpha}^{[n+1]}$.
This ``look-ahead'' feature is called \emph{extragradient}.

Lastly, we put the Bregman distance $\mathbb{KL}$ in \eqref{W}, the function $\Xi_m$ in \eqref{Xi}, the gradient in \eqref{gradient}, and a proper $\eta$ into \eqref{BFOM1}--\eqref{BFOM4}.
Based on the KKT conditions, the equations \eqref{BFOM1}--\eqref{BFOM2} are shown to be equivalent to
\begin{subequations}
\begin{align}
p^{\diamond}_k&=\left\{
\sum_{i=1}^K
p^{[n]}_i\mathrm{exp}\left[
-\eta\left(\mathop{\sum}_{m=1}^{M}\alpha_m^{[n]}\beta_m\nabla_{p_i}\Xi_m(\mathbf{p}^{[n]})\right)
\right]\right\}^{-1}
 \nonumber
\\
&\quad{}
\times
P\, p^{[n]}_k\mathrm{exp}\left[
-\eta\left(\mathop{\sum}_{m=1}^{M}\alpha_m^{[n]}\beta_m\nabla_{p_k}\Xi_m(\mathbf{p}^{[n]})\right)
\right]
,
\nonumber\\
&\quad{}
k=1,\cdots,K,
\label{BFOM21}
\\
\alpha^{\diamond}_m&=
\left(
\sum_{i=m}^M
\alpha^{[n]}_i\mathrm{exp}\left[
\eta\,\beta_i
\Xi_i(\mathbf{p}^{[n]})
\right]\right)^{-1}
\nonumber\\
&\quad{}
\times
\alpha^{[n]}_m\mathrm{exp}\left[
\eta\,\beta_m
\Xi_m(\mathbf{p}^{[n]})
\right]
,\quad
m=1,\cdots,M.
\label{BFOM22}
\end{align}
\end{subequations}
The equations \eqref{BFOM3}--\eqref{BFOM4} can be similarly reduced to closed-form expressions.

According to \textbf{Proposition 3} and \cite{BFOM1}, the mirror-prox LCPA algorithm is guaranteed to converge to the optimal solution to $\mathrm{P}4$.
But in practice, we can terminate the iterative procedure when the norm $\left\Vert\mathbf{p}^{[n]}-\mathbf{p}^{[n-1]}\right\Vert_{\infty}$ is small enough, e.g., $\left\Vert\mathbf{p}^{[n]}-\mathbf{p}^{[n-1]}\right\Vert_{\infty}<10^{-8}$.
The entire procedure for computing the solution to $\mathrm{P}4$ using the mirror-prox LCPA is summarized in the right hand branch of Fig.~4c.
In terms of computational complexity, computing \eqref{BFOM21} requires a complexity of $\mathcal{O}(MK)$.
Since the number of iterations for the mirror-prox LCPA to converge is $\mathcal{O}(1/\epsilon)$ with $\epsilon$ being the target solution accuracy, the total complexity of mirror-prox LCPA would be $\mathcal{O}(MK/\epsilon)$.

\section{Simulation Results and Discussions}

This section provides simulation results to evaluate the performance of the proposed algorithms.
It is assumed that the noise power $\sigma^2=-87~\mathrm{dBm}$  (corresponding to power spectral density $-140~\mathrm{dBm/Hz}$ with $180~\mathrm{kHz}$ bandwidth \cite{iot3}), which includes thermal noise and receiver noise.
The total transmit power at users is set to $P=13~\mathrm{dBm}$ (i.e., $20~\rm{mW}$), with the communication bandwidth $B=180~\mathrm{kHz}$.
The path loss of the $k^{\mathrm{th}}$ user $\varrho_{k}=-100~\rm{dB}$ is adopted \cite{wang}, and $\mathbf{h}_{k}$ is generated according to $\mathcal{CN}(\mathbf{0},\varrho_{k}\mathbf{I}_N)$ \cite{massive2}.
Without otherwise specified, it is assumed that $\xi=1$.
We set $Z_k^{\mathrm{min}} =0$ and $Z_k^{\mathrm{max}}=+\infty$ for all $k$ in Sections VII-A to VII-C.
Simulations with upper and lower bounds on data amount will be presented in Section VII-D.
Each point in the figures is obtained by averaging over $10$ simulation runs, with independent channels in each run.
All optimization problems are solved by Matlab R2015b on a desktop with Intel Core i5-4570 CPU at 3.2~GHz and 8~GB RAM.
All the classifiers are trained by Python 3.6 on a GPU server with Intel Core i7-6800 CPU at 3.4 GHz and GeForce GTX 1080 GPU.

\subsection{CNN and SVM}

We consider the aforementioned CNN and SVM classifiers with the number of learning models $M=2$ at the edge: i) Classification of MNIST dataset \cite{MNIST} via CNN; ii) Classification of digits dataset in Scikit-learn \cite{sklearn} via SVM.
The data size of each sample is $D_1=6276~\mathrm{bits}$ for the MNIST dataset and $D_2=324~\mathrm{bits}$ for the digits dataset in Scikit-learn.
It is assumed that there are $A_1=300$ CNN samples and $A_2=200$ SVM samples in the historical dataset.
The parameters in the two classification error models are assumed to be perfectly known and they are given by $(a_1,b_1)=(7.3,0.69)$ for CNN and $(a_2,b_2)=(5.2,0.72)$ for SVM as in Fig.~3b.
Finally, it is assumed that $(\beta_1,\beta_2)=(1,1.2)$ since the approximation error of SVM in Fig.~3b is larger than that of CNN.

To begin with, the case of $N=20$ and $K=4$ with $\mathcal{Y}_1=\{1\}$ and $\mathcal{Y}_2=\{2,3,4\}$ is simulated.
Under the above settings, we compute the collected sample sizes by executing the proposed MM-based LCPA, and the maximum error of classifiers (i.e., the worse classification result of CNN and SVM, where the classification error for each task is obtained from the machine learning experiment using the sample sizes from the power allocation algorithms) versus the total transmission time $T$ is shown in Fig. 5a.
Besides the proposed MM-based LCPA, we also simulate two benchmark schemes: 1) Max-min fairness scheme \cite[Sec. II-C]{fair}; 2) Sum-rate maximization scheme \cite[Sec. IV]{sumrate}.
It can be seen from Fig.~5a that the proposed MM-based LCPA algorithm with $\mathcal{I}=10$ has a significantly smaller classification error compared to other schemes, and the gap concisely quantifies the benefit brought by more training images for CNN under joint communication and learning design.
For example, at $T=20$ in Fig.~5a, the proposed MM-based LCPA collects $2817$ MNIST images on average, while the sum-rate maximization and the max-min fairness schemes collect $1686$ images and $1781$ images, respectively.
Furthermore, if we target at the same learning error, the proposed algorithm saves the transmission time by at least $30\%$ compared to benchmark schemes.
This can be seen from Fig.~5a at the target error $4.5\%$, where the proposed algorithm takes $10$ seconds for transmission, but other methods require about $20$ seconds.
The saved time enables the edge to collect data for other edge computing tasks \cite{mec}.

\begin{figure*}[!t]
 \centering
  \subfigure[]{
    \includegraphics[width=58mm]{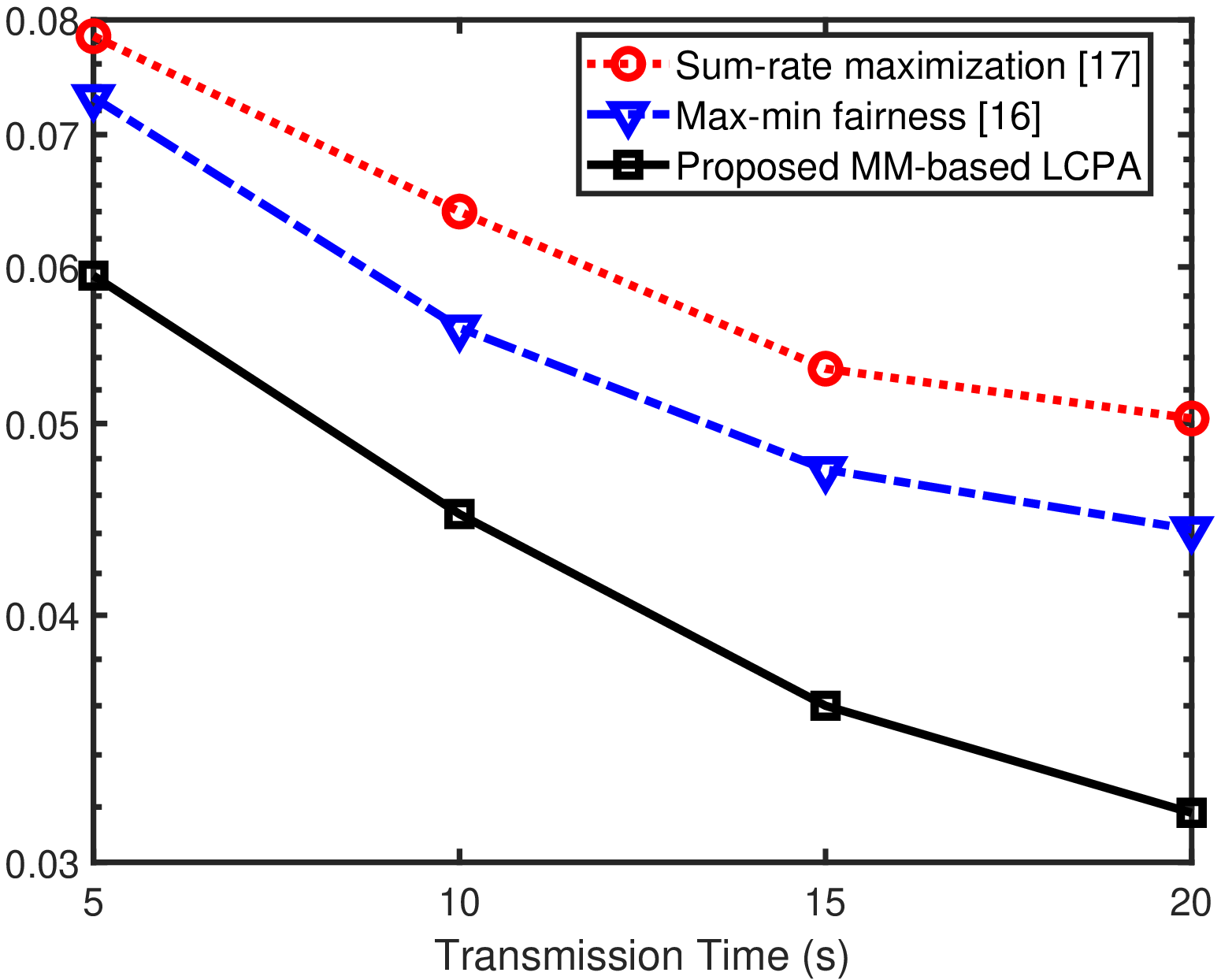}}
      \subfigure[]{
    \includegraphics[width=58mm]{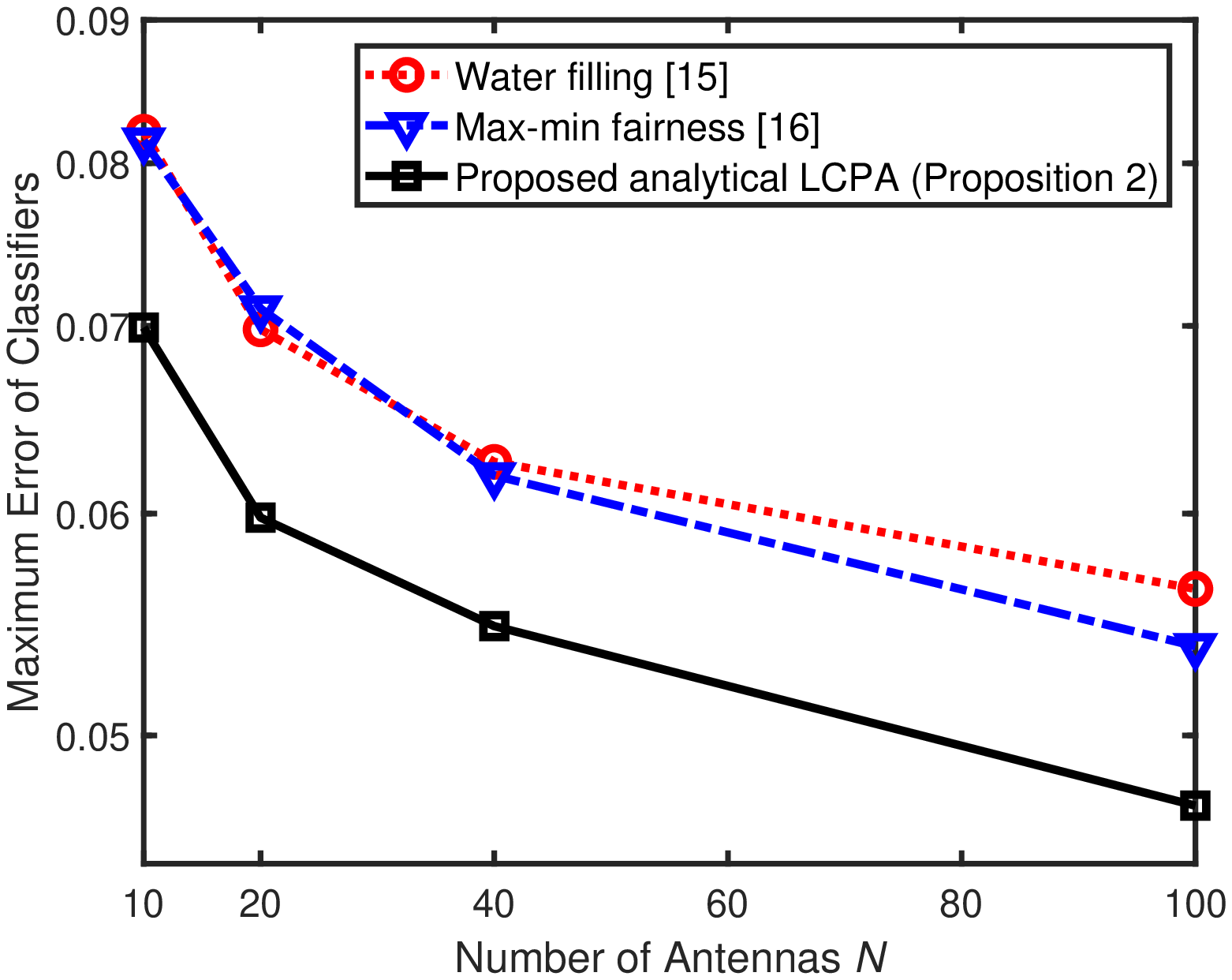}}
      \subfigure[]{
    \includegraphics[width=58mm]{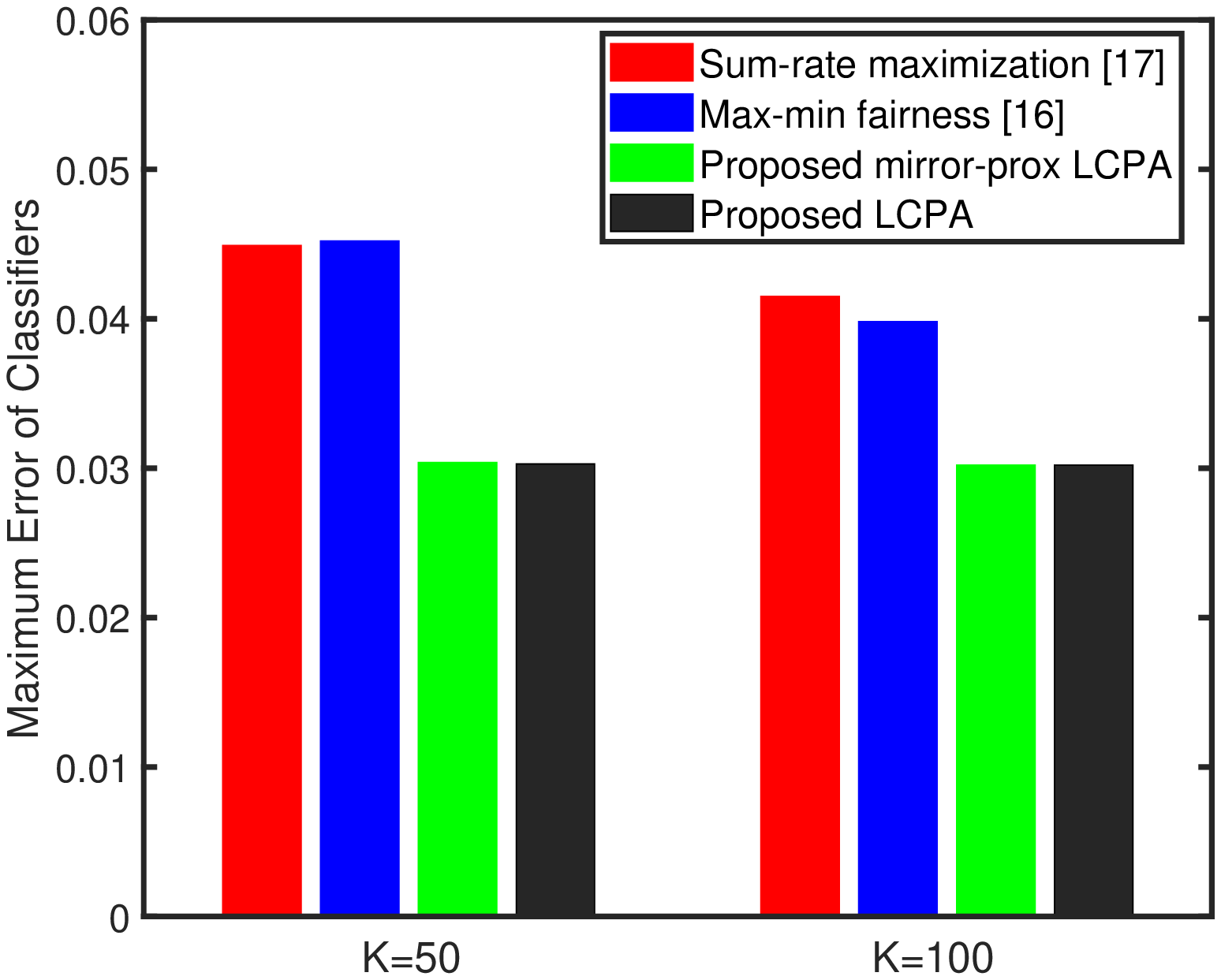}}
  \caption{a) Maximum error of classifiers versus total transmission time $T$ when $N=20$ and $K=4$; b) Maximum error of classifiers versus the number of antennas $N$ when $M=K=2$; c) Maximum error of classifiers versus the number of users $K$ when $N=100$;
}
\end{figure*}

To get more insight into the edge learning system, the case of $K=2$ with $\mathcal{Y}_1=\{1\}$ and $\mathcal{Y}_2=\{2\}$ at $T=5~\mathrm{s}$ is simulated, and the classification error versus the number of antennas $N=\{10,20,40,100\}$ is shown in Fig.~5b.
It can be seen from Fig.~5b that the classification error decreases as the number of antennas increases, which demonstrates the advantage of employing massive MIMO in edge machine learning.
More importantly, the proposed analytical solution in \textbf{Proposition 2} outperforms the water-filling\footnote{In the case of large $N$, sum-rate maximization scheme \cite[Sec. IV]{sumrate} would reduce to the iterative water-filling scheme  \cite{waterfilling}, which allocates power according to \eqref{water}.} and max-min fairness schemes even at a relatively small number of antennas $N=10$.
This is achieved by allocating much more power resources to the first user (i.e., user uploading datasets for CNN), because training CNN is more difficult than training SVM.
In particular, the transmit powers in $\rm{mW}$ are given by: 1) $(p_1,p_2)=(19.8473, 0.1524)$ for the analytical LCPA scheme; 2) $(p_1,p_2)=(9.9862, 10.0138)$ for the water-filling scheme; and 3) $(p_1,p_2)=(10.0869, 9.9131)$ for the max-min fairness scheme.
Notice that the performance gain brought by LCPA in Fig.~5b is slightly smaller than that in Fig.~5a, since the ratio $|\mathcal{Y}_1|/|\mathcal{Y}_2|$ is increased. But no matter what value $|\mathcal{Y}_1|$ and $|\mathcal{Y}_2|$ take, the proposed LCPA always outperforms existing algorithms due to its learning centric feature.

To verify the performance and the low complexity nature of the mirror-prox LCPA in Section VI when the number of antennas is large, the case of $N=100$
and $K\in\{50,100\}$ at $T=5~\mathrm{s}$ is simulated, with $\mathcal{Y}_1$ containing the first $1/5$ users and $\mathcal{Y}_2$ containing the rest $4/5$ users.
The maximum error of classifiers versus the number of users $K$ is shown in Fig.~5c.
It can be seen that the proposed mirror-prox LCPA algorithm significantly reduces the classification error compared to the water-filling and max-min fairness schemes, and it achieves performance close to that of the MM-based LCPA.
On the other hand, the average execution time at $K=100$ is given by: 1) $18.2496~\rm{s}$ for the MM-based LCPA; 2) $0.5673~\rm{s}$ for the mirror-prox LCPA; 3) $0.0044~\rm{s}$ for the water-filling scheme; and 4) $0.0054~\rm{s}$ for the max-min fairness scheme.
Compared with MM-based LCPA, the mirror-prox LCPA saves at least $95\%$ of the computation time, which corroborates the linear complexity derived in Section VI.

\subsection{Deep Neural Networks}

Next, we consider the ResNet-$110$ as task $1$ and the CNN in Section III as task $2$ at the edge, with $D_1=24584~\mathrm{bits}$ and $D_2=6276~\mathrm{bits}$.
The error rate parameters are given by $(a_1,b_1)=(8.15,0.44)$ and $(a_2,b_2)=(7.3,0.69)$.
In addition, it is assumed that $(\beta_1,\beta_2)=(1,1)$, $T=200~\mathrm{s}$, and there is no historical sample at the edge (i.e., $A_1=A_2=0$).
We simulate the case of $N=20$ and $K=4$ with $\mathcal{Y}_1=\{1,2\}$ and $\mathcal{Y}_2=\{3,4\}$.
For ResNet-$110$, we assume that the datasets $\{\mathcal{D}_1,\mathcal{D}_2\}$ are formed by dividing the CIFAR-10 dataset into two parts, each with $30000$ different samples.
For CNN, we assume that the datasets $\{\mathcal{D}_3,\mathcal{D}_4\}$ are formed by dividing the MNIST dataset into two parts, each also with $30000$ different samples.
The worse classification error between the two tasks (obtained from the machine learning experiment using the sample sizes from the power allocation algorithms) is: 1) $14.13\%$ for MM-based LCPA; 2) $16.79\%$ for the sum-rate maximization scheme; and 3) $16.42\%$ for the max-min fairness scheme.
It can be seen that the proposed LCPA achieves the smallest classification error, which demonstrates the effectiveness of learning centric resource allocation under deep neural networks and large datasets.

\begin{figure*}[!t]
 \centering
  \subfigure[]{
    \includegraphics[width=58mm]{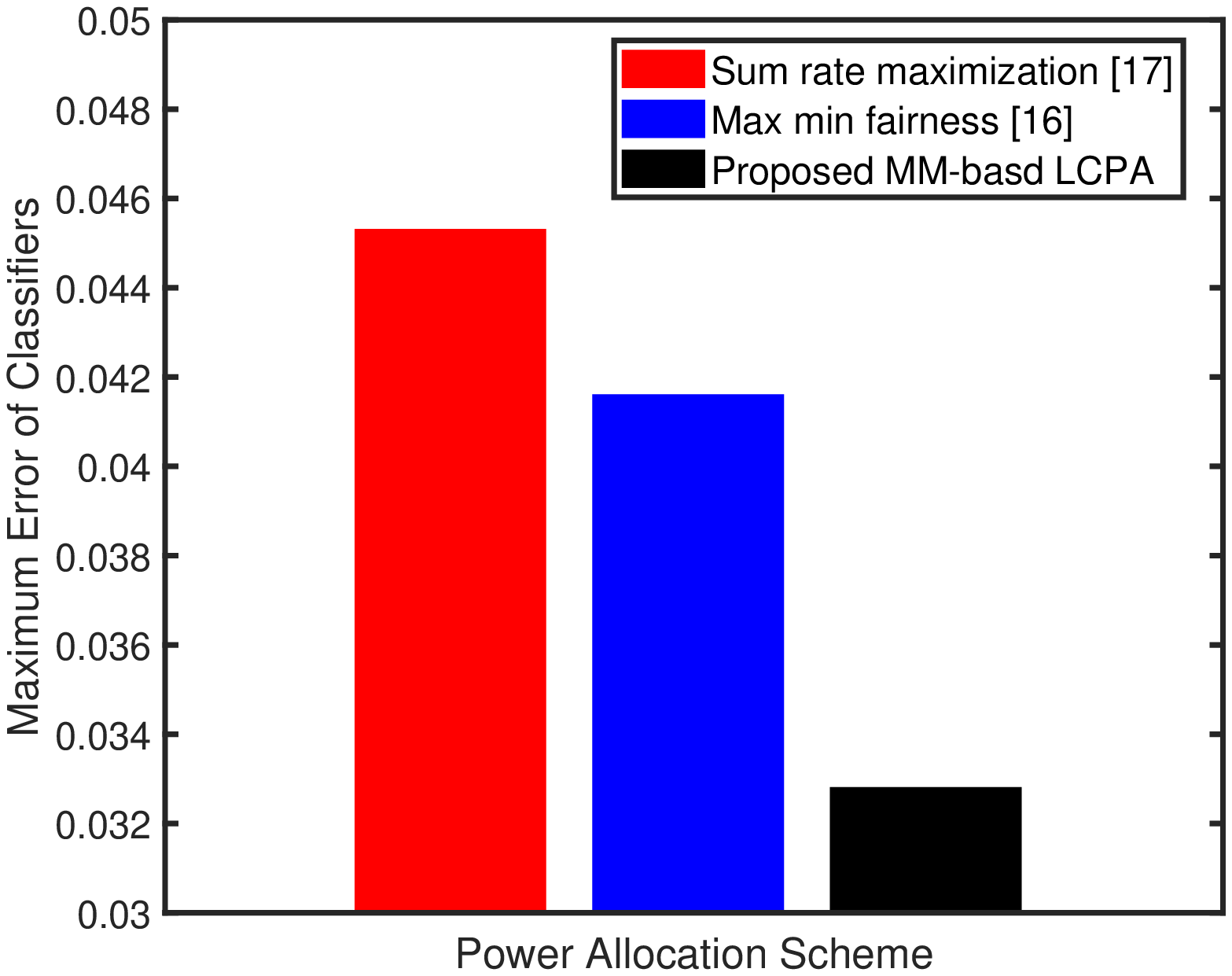}}
      \subfigure[]{
    \includegraphics[width=58mm]{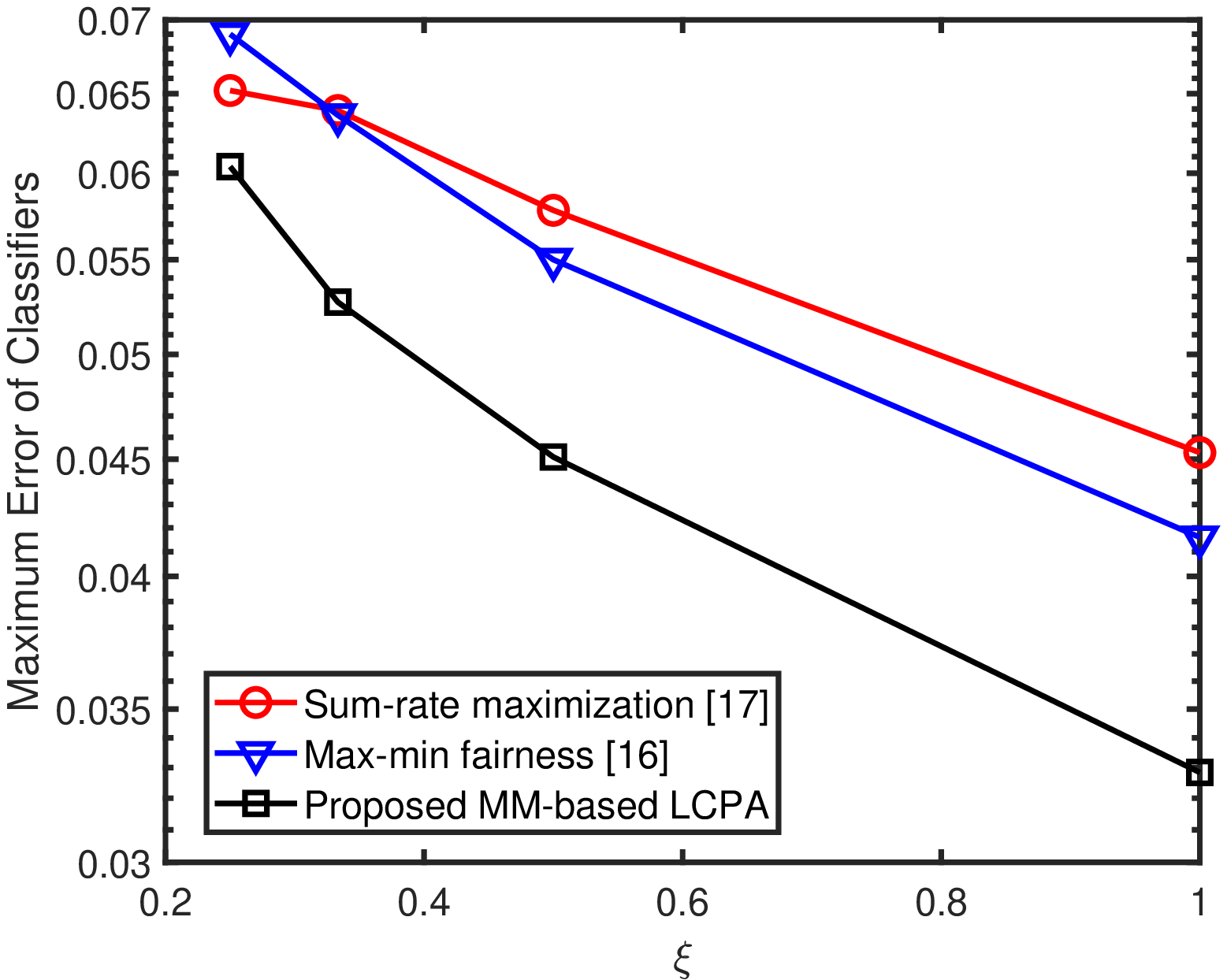}}
      \subfigure[]{
    \includegraphics[width=58mm]{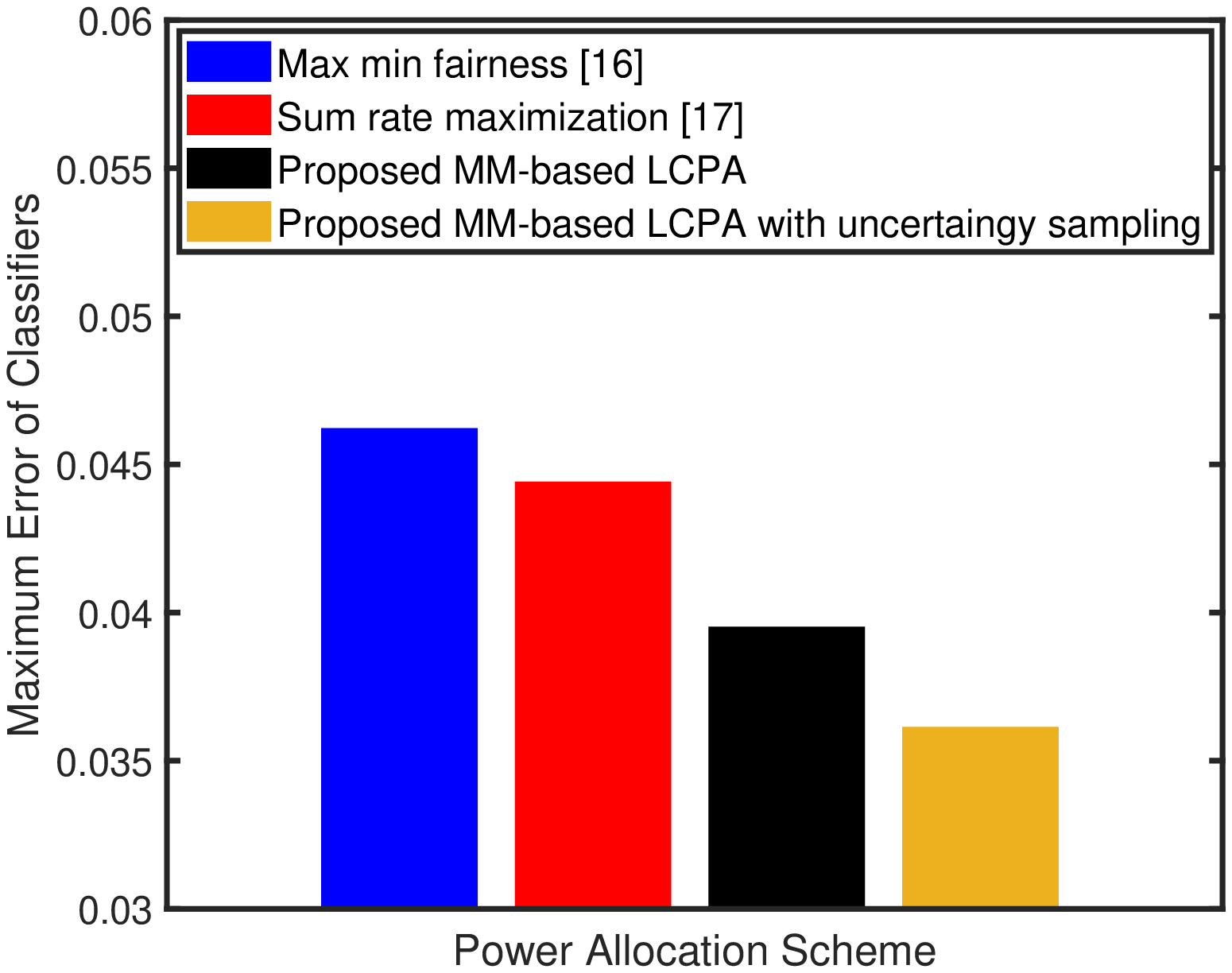}}
  \caption{a) Comparison of classification error when $N=100$ and $K=5$ at $T=10~\rm{s}$ with $(a_m,b_m)$ estimated from historical data; b) Maximum error of classifiers versus $\xi$ when $N=100$ and $K=5$; c) Comparison of classification error when $N=100$ and $K=6$ at $T=10~\rm{s}$ in the non-IID case.
}
\end{figure*}

\subsection{Practical Considerations}

In Sections VII-A and VII-B, we have assumed that the parameters $(a_m,b_m)$ are perfectly known.
In practice, we need to estimate them from the historical data.
More specifically, denote the historical data as $\mathcal{H}_k\subseteq \mathcal{D}_k$.
We take $2/3$ of the historical data in $\mathcal{H}_k$ as the training dataset $\mathcal{T}_k$, and the remaining $1/3$ as the validation dataset $\mathcal{V}_k$ for all $k$.
Therefore, the number of training samples for estimating $(a_m,b_m)$ is $\sum_{k\in\mathcal{Y}_m}|\mathcal{T}_k|$.
Based on the above dataset partitioning, the parameters $\{a_m, b_m\}$ are estimated in two stages:
\begin{itemize}
\item In the training stage, we can vary the sample size $v_m$ as $(v_m^{(1)},v_m^{(2)},\cdots)$ within $[0,\sum_{k\in\mathcal{Y}_m}|\mathcal{T}_k|]$.
For each $v_m^{(i)}$, we train the learning model for $E$ epochs and test the trained model on the validation dataset $\{\mathcal{V}_k\}_{k\in\mathcal{Y}_m}$.
The classification errors corresponding to the different sample sizes are given by $(\Psi_m(v_m^{(1)}),\Psi_m(v_m^{(2)}),\cdots)$.

\item In the fitting stage, with the classification error versus sample size, the parameters $(a_m,b_m)$ in $\Theta_m$ are found via the nonlinear least squares fitting in \eqref{fitting}.
\end{itemize}
Assuming that the complexity of processing each image (e.g., for CNN, each processing stage includes a backward pass and a forward pass) is $\mathcal{O}(W_m)$ where $W_m$ is the number of parameters in the learning model, the complexity in the training stage is $\mathcal{O}(EW_m\sum_{i}v_m^{(i)})$.
If $\sum_{i}v_m^{(i)}$ is smaller than the total number of samples after data collection, then $\mathcal{O}(EW_m\sum_{i}v_m^{(i)})$ would be smaller than the actual model training complexity.
On the other hand, the complexity in the fitting stage is negligible, since the problem \eqref{fitting} only has two variables.
Even using a naive brute-force search with a step size of $0.001$ for each dimension, the running time for solving \eqref{fitting} is only $5$ seconds on a desktop with Intel Core i5-4570 CPU at 3.2~GHz and 8~GB RAM.

To illustrate the above procedure, we consider the case of $N=100$ and $K=5$ with $M=2$, where the first task is to train the CNN with user $\mathcal{Y}_1=\{1\}$ and the second task is to train the SVM with users $\mathcal{Y}_1=\{2,3,4,5\}$.
For CNN, we assume that the dataset $\mathcal{D}_1$ contains $10000$ different samples from the MINST dataset.
For SVM, we assume that the datasets $\{\mathcal{D}_2,\mathcal{D}_3,\mathcal{D}_4,\mathcal{D}_5\}$ are formed by dividing the scikit-learn digit dataset into four parts, each with $250$ different samples.
To conform with the size of $\mathcal{D}_1$, the datasets $\{\mathcal{D}_2,\mathcal{D}_3,\mathcal{D}_4,\mathcal{D}_5\}$ are augmented by appending $9750$ samples (generated via random rotations of the original samples) to the $250$ original samples.
The test dataset for CNN contains another $1000$ samples from the MNIST dataset and the test dataset for SVM contains another $797$ samples from the scikit-learn digits dataset.
It is assumed that $|\mathcal{H}_1|=450$ and $|\mathcal{H}_2|=\cdots=|\mathcal{H}_5|=75$.
According to the ``$2/3$-training and $1/3$-validation'' partitioning of $\mathcal{H}_k$, we have $|\mathcal{T}_1|=300$ and $|\mathcal{T}_2|=\cdots=|\mathcal{T}_5|=50$, meaning that there are $300$ and $200$ historical samples for training CNN and SVM, respectively.
For CNN, we vary the sample size $v_1$ as $(v_1^{(1)},v_1^{(2)},\cdots)=(100,150,200,300)$ and perform training for $E=200$ epochs.
The classification errors on the validation dataset $\mathcal{V}_1$ are given by $(0.300, 0.200, 0.140, 0.120)$, and after model fitting we have $(a_1,b_1)=(9.74,0.77)$.
On the other hand, for SVM, we vary the sample size $v_2$ as $(v_2^{(1)},v_2^{(2)},\cdots)=(30,50,100,200)$.
The classification errors on the validation dataset $\{\mathcal{V}_2,\mathcal{V}_3,\mathcal{V}_4,\mathcal{V}_5\}$ are given by $(0.510, 0.280, 0.220, 0.050)$, and after model fitting we have $(a_2,b_2)=(14.27,0.85)$.

Based on the above estimation results for $(a_m,b_m)$, the classification errors at $T=10~\mathrm{s}$ for different schemes are compared in Fig.~6a.
It can be seen from Fig.~6a that the proposed LCPA reduces the minimum classification error by at least $20\%$ compared with other simulated schemes.
To demonstrate the performance of the proposed LCPA under various values of $\xi$, the classification error versus $\xi=\{1/4,1/3,1/2,1\}$ at $T=10~\mathrm{s}$ is shown in Fig.~6b.
It can be seen from Fig.~6b that the classification error increases as $\xi$ decreases, which is due to the loss of samples during wireless transmission.
However, no matter what value $\xi$ takes, the proposed LCPA achieves the minimum classification error among all the simulated schemes.

\subsection{Non-IID Dataset}

In the non-IID case, we add one more user (denoted as user 6) to the system considered in Section VII-C, with the first task training the CNN with users $\mathcal{Y}_1=\{1,6\}$ and the second task training the SVM with users $\mathcal{Y}_1=\{2,3,4,5\}$.
The dataset $\mathcal{D}_6$ repeats the first sample in the MINST dataset for $10000$ times and $|\mathcal{H}_6|=450$.
Out of this $\mathcal{H}_6$, we use $2/3$ as the training dataset $\mathcal{T}_6$ and $1/3$ as the validation dataset $\mathcal{V}_6$.
We train the CNN in Fig.~1 using the training dataset $\{\mathcal{T}_1,\mathcal{T}_6\}$ for $200$ epochs, and obtain the parameter $\mathbf{w}$ in the CNN.
We then compute the prediction confidence for all the data $\mathbf{d}\in\{\mathcal{V}_1,\mathcal{V}_6\}$ based on equation \eqref{uncertain} with the probability function being the softmax output of CNN.
It turns out that the least confident samples in $\mathcal{V}_1$ and $\mathcal{V}_6$ have the following scores: $\mathop{\mathrm{min}}_{\mathbf{d}\in\mathcal{V}_1}U(\mathbf{d})=0.274$ and $\mathop{\mathrm{min}}_{\mathbf{d}\in\mathcal{V}_6}U(\mathbf{d})=0.999$.
For user $1$, since the CNN model is not sure about its prediction, we set $Z_1^{\mathrm{min}}=100$ and $Z_1^{\mathrm{max}}=10000$.
For user $6$, since the CNN model can predict its validation data with very high confidence, we set $Z_6^{\mathrm{min}}=0$ and $Z_6^{\mathrm{max}}=10$.
Based on the ratio between $Z_1^{\mathrm{max}}$ and $Z_6^{\mathrm{max}}$, we use $450$ samples from $\mathcal{H}_1$ and $10/10000\times450$ samples from $\mathcal{H}_6$ for estimating $(a_m,b_m)$.
Then we have $(a_1,b_1)=(9.74,0.77)$.
For the SVM model, since the data from users $\{2,3,4,5\}$ are assumed to be IID, $(a_2,b_2)$ is the same as that in Section VII-C, i.e., $(a_2,b_2)=(14.27,0.85)$.

The comparison of classification error in the non-IID case is shown in Fig.~6c.
It can be seen from Fig.~6c that the performance of MM-based LCPA (without uncertainty sampling) becomes worse than that in Fig.~6a, which is due to the additional resources allocated to user $6$.
However, the proposed MM-based LCPA still outperforms the sum-rate maximization and max-min fairness schemes.
More importantly, by adopting the uncertainty sampling method to distinguish the users' data quality, the classification error of LCPA can be further reduced (e.g., from the black bar to the yellow bar in Fig.~6c).
This demonstrates the benefit brought by \emph{joint estimation of sample size and prediction confidence} in edge machine learning systems.

\section{Conclusions}

This paper has introduced the LCPA concept to edge machine learning.
By adopting an empirical classification error model, learning efficient edge resource allocation has been obtained via the MM-based LCPA algorithm.
In the large-scale settings, a fast FOM has been derived to tackle the curse of high-dimensionality.
Simulation results have shown that the proposed LCPA algorithms achieve lower classification errors than existing power allocation schemes.
Furthermore, the proposed fast algorithm significantly reduces the execution time compared with the MM-based LCPA while still achieving satisfactory performance.

\appendices
\section{Proof of Proposition 1}
To prove part (i), consider the following inequality:
\begin{align}
&-\mathrm{ln}\left(\sum_{l=1,l\neq k}^K\frac{G_{k,l}p^\star_{l}}{\sigma^2}+1\right)
-\frac{\sum_{l=1,l\neq k}^KG_{k,l}p_{l}/\sigma^2+1}{\sum_{l=1,l\neq k}^KG_{k,l}p^\star_{l}/\sigma^2+1}
+1
\nonumber\\
&
\leq
-
\mathrm{ln}\left(\sum_{l=1,l\neq k}^K\frac{G_{k,l}p_{l}}{\sigma^2}+1
\right), \label{A1}
\end{align}
which is obtained from $-\mathrm{ln}(x')-\frac{1}{x'}(x-x')\leq-\mathrm{ln}(x)$ for any $(x,x')$ due to the convexity of $-\mathrm{ln}(x)$.
Adding $\mathrm{ln}\left(\sum_{l=1}^KG_{k,l}p_{l}/\sigma^2+1 \right)$ on both sides of \eqref{A1}, we obtain
\begin{align}
&\mathrm{ln}\left(\sum_{l=1}^K\frac{G_{k,l}p_{l}}{\sigma^2 }+1
\right)
-
\mathrm{ln}\left(\sum_{l=1,l\neq k}^K\frac{G_{k,l}p^\star_{l}}{\sigma^2}+1\right)
\nonumber\\
&
-\frac{\sum_{l=1,l\neq k}^KG_{k,l}p_{l}/\sigma^2+1}{\sum_{l=1,l\neq k}^KG_{k,l}p^\star_{l}/\sigma^2+1}
+1
\nonumber
\\
&
\leq
\mathrm{ln}\left(\sum_{l=1}^K\frac{G_{k,l}p_{l}}{\sigma^2}+1\right)
-
\mathrm{ln}\left(\sum_{l=1,l\neq k}^K\frac{G_{k,l}p_{l}}{\sigma^2}+1\right)
\nonumber
\\
&
=\mathrm{ln}\left(1+\frac{G_{k,k}p_{k}}{\sum_{l=1,l\neq k}^KG_{k,l}p_{l}+
\sigma^2}\right). \label{A2}
\end{align}
Putting the result of \eqref{A2} into $\widetilde{\Phi}_{m}(\mathbf{p}|\mathbf{p}^\star)$ in \eqref{Phi} and since $a_mx^{-b_m}$ is a decreasing function of $x$, we immediately prove
\begin{align}
&\widetilde{\Phi}_{m}(\mathbf{p}|\mathbf{p}^\star)
\nonumber\\
&\geq a_m\Bigg[
\sum_{k\in\mathcal{Y}_m} \frac{\xi BT}{D_m\mathrm{ln}2\,}
\mathrm{ln}\left(1+\frac{G_{k,k}p_{k}}{\sum_{l\neq k}G_{k,l}p_{l}+
\sigma^2}\right)+A_m
\Bigg]^{-b_m}
\nonumber\\
&
=
\Phi_{m}(\mathbf{p}).
\end{align}

To prove part (ii), we first notice that $\widetilde{\Phi}_m(\mathbf{p}|\mathbf{p}^{\star})=h_m\left( g_m(\mathbf{p}|\mathbf{p}^\star)\right)$ is a composition function of $h_m\circ g_m$, where
$h_m(x)=a_mx^{-b_m}$ and
\begin{align}
g_m(\mathbf{p}|\mathbf{p}^\star)
&=
\sum_{k\in\mathcal{Y}_m} \frac{\xi BT}{D_m\mathrm{ln}2\,}
\Bigg[
\mathrm{ln}\left(\sum_{l=1}^K\frac{G_{k,l}p_{l}}{\sigma^2}+1\right)
\nonumber\\
&\quad{}
-
\mathrm{ln}\left(\sum_{l=1,l\neq k}^K\frac{G_{k,l}p^\star_{l}}{\sigma^2 }+1\right)
\nonumber
\\&
\quad{}
-\frac{\sum_{l=1,l\neq k}^KG_{k,l}p_{l}/\sigma^2+1}{\sum_{l=1,l\neq k}^KG_{k,l}p^\star_{l}/\sigma^2+1}
+1
\Bigg]
+A_m. \label{B1}
\end{align}
Since $\nabla h_m(x)=-a_mb_mx^{-b_m-1}\leq 0$ and $\nabla^2 h_m(x)=a_mb_m(b_m+1)x^{-b_m-2}\geq0$, the function $h_m(x)$ is convex and nonincreasing.
Adding to the fact that $g_m(\mathbf{p}|\mathbf{p}^\star)$ is a concave function of $\mathbf{p}$, we immediately prove the convexity of $\widetilde{\Phi}_m$ using the composition rule \cite[Ch. 3, pp. 84]{opt1}.

Finally, to prove $\widetilde{\Phi}_{m}(\mathbf{p}^\star|\mathbf{p}^\star)=\Phi_{m}(\mathbf{p}^\star)$ and $\nabla_{\mathbf{p}}\widetilde{\Phi}_{m}(\mathbf{p}^\star|\mathbf{p}^\star)=\nabla_{\mathbf{p}}\Phi_{m}(\mathbf{p}^\star)$, we put $\mathbf{p}=\mathbf{p}^\star$ into the functions $\{\widetilde{\Phi}_{m},\Phi_m,\nabla_{\mathbf{p}}\widetilde{\Phi}_{m},\nabla_{\mathbf{p}}\Phi_{m}\}$, and we immediately obtain part (iii).

\section{Proof of Proposition 2}

To prove this proposition, the Lagrangian of $\mathrm{P}2$ is
\begin{align}
L&=\mu+\sum_{k=1}^K\nu_k\Bigg[\beta_ka_k\left(
\frac{\xi BT}{D_k}\mathrm{log}_2\left(1+\frac{G_{k,k}p_{k}}{\sigma^2} \right)+A_k
\right)^{-b_k}
\nonumber\\
&\quad{}
-\mu\Bigg]
+\chi\left(\mathop{\sum}_{k=1}^Kp_k-P\right)-\sum_{k=1}^K\theta_kp_k, \label{C0}
\end{align}
where $\{\nu_k,\chi,\theta_k\}$ are non-negative Lagrange multipliers.
According to the KKT conditions $\partial L/\partial \mu^*=0$ and $\partial L/\partial p_k^*=0$ \cite{opt1}, the optimal $\{\mu^*,\,p_k^*,\,\nu_k^*,\,\chi^*,\theta_k^*\}$ must together satisfy
\begin{align}
&1-\sum_{k=1}^K\nu_k^*=0,\quad
\chi^*+\nu_k^*\times F_k(p_k^*)
=\theta_k^*,\quad \forall k, \label{C1}
\end{align}
where
\begin{align}
F_k(x)&=-\frac{\beta_ka_kb_k\xi BT}{D_k\mathrm{ln}2\,(\sigma^2G_{k,k}^{-1}+x)}
\nonumber\\
&\quad{}
\times
\left[\frac{\xi BT}{D_k}\mathrm{log}_2\left(1+\frac{G_{k,k}x}{\sigma^2} \right)+A_k
\right]^{-b_k-1}
,
\end{align}
with $x\geq 0$ ($x\neq0$ if $A_k=0$).
Notice that $F_k(x)<0$ holds for any $x\geq 0$.
Based on the result of \eqref{C1}, it is clear that
$\sum_{k=1}^K\frac{\theta_k^*-\chi^*}{F_k(p_k^*)}=1$.
Adding to the fact that $\theta_k^*\geq 0$ and $F_k(p_k^*)<0$, we must have $\chi^*\neq 0$.
Now we will consider two cases.
\begin{itemize}
\item $p_k^*=0$. In this case, $\beta_ka_ku_k^{-b_k}\leq \mu^*$ must hold.

\item $p_k^*>0$.
In such a case, based on the complementary slackness condition, we must have $\theta_k^*=0$.
Putting $\theta_k^*=0$ into \eqref{C1} and using $\chi^*\neq 0$, $\nu_k^*\neq0$ holds.
Using $\nu_k^*\neq0$ and the complementary slackness condition, $\beta_ka_k\left(
\frac{\xi BT}{D_k}\mathrm{log}_2\left(1+\frac{G_{k,k}p_{k}^*}{\sigma^2} \right)+A_k
\right)^{-b_k}=\mu^*$ for all $k$.
\end{itemize}
Combining the above two cases gives \eqref{pk*} and  the proposition is proved.

\section{Proof of Proposition 3}
To prove this proposition, we need the following lemma for $\nabla_{\mathbf{p}}\Xi_m(\mathbf{p})$.
\begin{lemma}
If $\mu\leq\mu_0$, the gradient function $\nabla_{\mathbf{p}}\Xi_m(\mathbf{p})$ satisfies the following:

\noindent(i) $\left\Vert\nabla_{\mathbf{p}}\Xi_m(\mathbf{p})\right\Vert_{2}\leq
L_{2}/\beta_m$;

\noindent(ii) $\left\Vert\nabla_{\mathbf{p}}\Xi_m(\mathbf{p})-\nabla_{\mathbf{p}}\Xi_m(\mathbf{p}')\right\Vert_{\infty}\leq
L_{1}/\beta_m\left\Vert\mathbf{p}-\mathbf{p}'\right\Vert_2$.
\end{lemma}
\begin{proof}
To begin with, the assumption $\mu=\mathrm{max}_m\beta_m\Xi_m(\mathbf{p})\leq \mu_0$ gives
\begin{align}
&\sum_{k\in\mathcal{Y}_m} \frac{\xi BT}{D_m}\mathrm{log}_2\left(1+\frac{G_{k,k}p_{k}}{\sigma^2} \right)+A_m
\geq\left(\frac{\beta_ma_m}{\mu_0}\right)^{\frac{1}{b_m}}. \label{lemma1}
\end{align}
Based on \eqref{lemma1} and the expression of $\partial \Xi_m/\partial p_j$ in \eqref{ximpartial}, we have
\begin{align}
\Big|\frac{\partial \Xi_m}{\partial p_j}\Big|
&\leq  a_mb_m\left(\frac{\beta_ma_m}{\mu_0}\right)^{-1-1/b_m}
\frac{\xi BT\mathbb{I}_{\mathcal{Y}_m}(j)}{D_m\mathrm{ln}2\,(\sigma^2G_{j,j}^{-1}+p_j)}
\nonumber \\
&\leq
\frac{a_mb_m\xi BTG_{j,j}\mathbb{I}_{\mathcal{Y}_m}(j)}{D_m\mathrm{ln}2\,\sigma^2}
\left(\frac{\mu_0}{\beta_ma_m}\right)^{1+1/b_m},
\end{align}
where the second inequality is due to $p_j\geq 0$.
Putting the above result into \eqref{gradient}, and based on the definition of $L_{2}$ in \eqref{L12}, part (i) is immediately proved.

Next, to prove part (ii),
we notice that the derivative in \eqref{ximpartial} can be rewritten as $\nabla_{p_j}\Xi_m(\mathbf{p})=h_m(\mathbf{p})\times g_{m,j}(\mathbf{p})$, with the auxiliary functions
\begin{subequations}
\begin{align}
h_m(\mathbf{p})&=
-a_mb_m\Bigg[
\sum_{k\in\mathcal{Y}_m} \frac{\xi BT}{D_m}\mathrm{log}_2\left(1+\frac{G_{k,k}p_{k}}{\sigma^2} \right)
\nonumber\\
&\quad{}
+A_m
\Bigg]^{-b_m-1},
\\
g_{m,j}(\mathbf{p})&=\frac{\xi BT\mathbb{I}_{\mathcal{Y}_m}(j)}{D_m\mathrm{ln}2\,\,(\sigma^2G_{j,j}^{-1}+p_j)}.
\end{align}
\end{subequations}
Using the result in \eqref{lemma1} and due to $p_j\geq 0$, we have
\begin{align}
|h_m(\mathbf{p})|&\leq
a_mb_m\left(\frac{\mu_0}{\beta_ma_m}\right)^{1+1/b_m},
\nonumber\\
|g_{m,j}(\mathbf{p})|&\leq \frac{\xi BTG_{j,j}\mathbb{I}_{\mathcal{Y}_m}(j)}{D_m\mathrm{ln}2\,\sigma^2}.
 \label{D3-1}
\end{align}
Furthermore, according to Lipschitz conditions \cite{BFOM2} of $h_m$ and $g_{m,j}$, they satisfy
\begin{subequations}
\begin{align}
|h_m(\mathbf{p})-h_m(\mathbf{p}')|
&\leq
\mathop{\mathrm{sup}}_{\mathbf{p}\in\mathcal{P}}\,\left\Vert\nabla_{\mathbf{p}}h_m(\mathbf{p})||_{2}\times||\mathbf{p}-\mathbf{p}'\right\Vert_2
\nonumber\\
&\leq
\frac{a_mb_m(b_m+1)\xi BTH_m}{D_m\mathrm{ln}2\,\sigma^2}
\nonumber\\
&\quad{}
\times
\left(\frac{\mu_0}{\beta_ma_m}\right)^{1+2/b_m}
\left\Vert\mathbf{p}-\mathbf{p}'\right\Vert_2, \label{D3-2a}
\\
|g_{m,j}(\mathbf{p})-g_{m,j}(\mathbf{p}')|&\leq
\mathop{\mathrm{sup}}_{\mathbf{p}\in\mathcal{P}}\,\left\Vert\nabla_{\mathbf{p}}g_{m,j}(\mathbf{p})\right\Vert_{2}\times\left\Vert\mathbf{p}-\mathbf{p}'\right\Vert_2
\nonumber\\
&
\leq
\frac{\xi BTG_{j,j}^2\mathbb{I}_{\mathcal{Y}_m}(j)}{D_m\mathrm{ln}2\,\sigma^4}\left\Vert\mathbf{p}-\mathbf{p}'\right\Vert_2. \label{D3-2b}
\end{align}
\end{subequations}
As a result, the following inequality is obtained:
\begin{align}
&|\nabla_{p_j}\Xi_m(\mathbf{p})-\nabla_{p_j}\Xi_m(\mathbf{p}')|
\nonumber\\
&
\leq
|h_m(\mathbf{p})|\,|g_{m,j}(\mathbf{p})-g_{m,j}(\mathbf{p}')|
\nonumber\\
&\quad{}
+|h_m(\mathbf{p})-h_m(\mathbf{p}')|\,|g_{m,j}(\mathbf{p}')| \nonumber
\\
&
\leq
\Bigg[
a_mb_m\left(\frac{\mu_0}{\beta_ma_m}\right)^{1+1/b_m}\frac{\xi BTG_{j,j}^2}{D_m\mathrm{ln}2\,\sigma^4}
+
\left(\frac{\mu_0}{\beta_ma_m}\right)^{1+2/b_m}
\nonumber\\
&\quad{}
\times
\frac{a_mb_m(b_m+1)\xi^2B^2T^2G_{j,j}H_m}{D_m^2\mathrm{ln}^22\sigma^4}\Bigg]
\left\Vert\mathbf{p}-\mathbf{p}'\right\Vert_2
\nonumber\\
&
\leq L_{1}/\beta_m\left\Vert\mathbf{p}-\mathbf{p}'\right\Vert_2
, \label{D4}
\end{align}
where the first inequality is due to $|ab+cd|\leq|a||b|+|c||d|$, and the second inequality is obtained from \eqref{D3-1} and \eqref{D3-2a}--\eqref{D3-2b}.
By taking the maximum of \eqref{D4} for all $j$, part (ii) is proved.
\end{proof}

Based on \textbf{Lemma 1}, we are now ready to prove the proposition.
In particular, according to \cite{BFOM1,BFOM2,BFOM3}, the function $\Upsilon(\bm{\alpha},\mathbf{p})$ is $(L_{1},L_{2},L_{2},0)$--smooth if and only if
\begin{subequations}
\begin{align}
&\left\Vert
\left[\mathop{\sum}_{m=1}^{M}\alpha_m\beta_m\nabla_\mathbf{p}\Xi_m(\mathbf{p})\right]
-
\left[\mathop{\sum}_{m=1}^{M}\alpha_m\beta_m\nabla_\mathbf{p}\Xi_m(\mathbf{p}')\right]
\right\Vert_{\infty}
\nonumber\\
&
\leq
L_{1}
\left\Vert\mathbf{p}-\mathbf{p}'\right\Vert_1,
\label{Lips1}
\\
&\left\Vert
\left[\mathop{\sum}_{m=1}^{M}\alpha_m\beta_m\nabla_\mathbf{p}\Xi_m(\mathbf{p})\right]
-
\left[\mathop{\sum}_{m=1}^{M}\alpha_m'\beta_m\nabla_\mathbf{p}\Xi_m(\mathbf{p})\right]
\right\Vert_{\infty}
\nonumber\\
&
\leq
L_{2}
\left\Vert\bm{\alpha}-\bm{\alpha}'\right\Vert_1, \label{Lips2}
\\
&\Big\Vert
\left[\beta_1\Xi_1(\mathbf{p}),\cdots,\beta_M\Xi_M(\mathbf{p}) \right]
\nonumber\\
&-\left[\beta_1\Xi_1(\mathbf{p}'),\cdots,\beta_M\Xi_M(\mathbf{p}') \right]
\Big\Vert_{\infty}
\leq
L_{2}
\left\Vert\mathbf{p}-\mathbf{p}'\right\Vert_1,
\label{Lips3}
\\
&\Big\Vert
\left[\beta_1\Xi_1(\mathbf{p}),\cdots,\beta_M\Xi_M(\mathbf{p}) \right]
\nonumber\\
&
-\left[\beta_1\Xi_1(\mathbf{p}),\cdots,\beta_M\Xi_M(\mathbf{p}) \right]
\Big\Vert_{\infty}
\leq
0\,
\left\Vert\bm{\alpha}-\bm{\alpha}'\right\Vert_1,
\label{Lips4}
\end{align}
\end{subequations}
for any $\mathbf{p}, \mathbf{p}'\in\mathcal{P}$ and $\bm{\alpha}, \bm{\alpha}'\in\mathcal{A}$.
To prove the above inequalities, we need to upper bound the left hand sides of \eqref{Lips1}--\eqref{Lips4} and compare the bounds with the right hand sides of \eqref{Lips1}--\eqref{Lips4}.
To this end, the left hand side of \eqref{Lips1} is bounded as
\begin{align}
&\left\Vert
\mathop{\sum}_{m=1}^{M}\alpha_m\beta_m\left[\nabla_\mathbf{p}\Xi_m(\mathbf{p})
-\nabla_\mathbf{p}\Xi_m(\mathbf{p}')\right]
\right\Vert_{\infty}
\nonumber\\
&\leq\mathop{\sum}_{m=1}^{M}\alpha_m\beta_m \left\Vert
\nabla_\mathbf{p}\Xi_m(\mathbf{p})
-\nabla_\mathbf{p}\Xi_m(\mathbf{p}')
\right\Vert_{\infty} \nonumber
\\
&
\leq
L_{1}\left\Vert\mathbf{p}-\mathbf{p}'\right\Vert_2,
\end{align}
where the last inequality is from \textbf{Lemma 1}.
Further due to $\left\Vert\mathbf{p}-\mathbf{p}'\right\Vert_2\leq\left\Vert\mathbf{p}-\mathbf{p}'\right\Vert_1$, the equation \eqref{Lips1} is proved.
On the other hand, the left hand side of \eqref{Lips2} is upper bounded as
\begin{align}
&\left\Vert
\mathop{\sum}_{m=1}^{M}(\alpha_m-\alpha_m')\beta_m\nabla_\mathbf{p}\Xi_m(\mathbf{p})
\right\Vert_{\infty}
\nonumber\\
&\leq\mathop{\sum}_{m=1}^{M}|\alpha_m-\alpha_m'| \left\Vert\beta_m
\nabla_\mathbf{p}\Xi_m(\mathbf{p})
\right\Vert_{\infty}
\nonumber\\
&
\leq
L_{2}
\left\Vert\bm{\alpha}-\bm{\alpha}'\right\Vert_1,
\end{align}
where the last inequality is from \textbf{Lemma 1}.
In addition, the left hand side of \eqref{Lips3} can be upper bounded as
\begin{align}
&\left\Vert
\left[\beta_1\Xi_1(\mathbf{p})-\beta_1\Xi_1(\mathbf{p}'),\cdots,\beta_M\Xi_M(\mathbf{p})-\beta_M\Xi_M(\mathbf{p}') \right]
\right\Vert_{\infty}
\nonumber\\
&\leq
\mathop{\mathrm{max}}_{m=1,\cdots,M}~\beta_m\times\mathop{\mathrm{sup}}_{\mathbf{p}\in\mathcal{P}}\,\left\Vert\nabla_{\mathbf{p}}\Xi_m(\mathbf{p})\right\Vert_{2}\times
\left\Vert\mathbf{p}-\mathbf{p}'\right\Vert_2
\nonumber\\
&\leq
L_{2}
\left\Vert\mathbf{p}-\mathbf{p}'\right\Vert_1.
\end{align}
Finally, since the left hand side of \eqref{Lips4} is zero, we immediately have that \eqref{Lips4} holds.

\end{document}